\documentclass[12pt, leqno]{article}
\usepackage[utf8]{inputenc}
\usepackage{setup_normal2}
\usepackage{newcommands}
\usepackage{titletoc}
\usepackage[most]{tcolorbox}

\dottedcontents{section}[4em]{}{2.9em}{0.7pc}
\dottedcontents{subsection}[0em]{}{3.3em}{1pc}

\title{{\textbf{\large{How Lagrangian states evolve into random waves}}}}
\author{Maxime Ingremeau\footnote{Laboratoire J.A. Dieudonné, Université Côte d'Azur. Contact: maxime.ingremeau@univ-cotedazur.fr} and Alejandro Rivera\footnote{\'Ecole Polytechnique Fédérale de Lausanne, Chair of Random Geometry. Contact : alejandro.rivera@epfl.ch }}
\date{}
\begin{document}
\thispagestyle{empty}
\maketitle

\begin{abstract}
In this paper, we consider a compact {\color{black}connected} manifold {\color{black}$(X,g)$} of negative curvature, and a family of semiclassical Lagrangian states $f_h(x) = a(x) e^{\frac{i}{h} \phi(x)}$ on $X$. For a wide family of phases $\phi$, we show that $f_h$, when evolved by the semiclassical Schrödinger equation during a long time, resembles a random Gaussian field. This can be seen as an analogue of Berry's random waves conjecture for Lagrangian states.
\end{abstract}

\newpage

\section{Introduction}

\paragraph{Berry's conjecture}

In his influential paper \cite{Ber}, M.V.Berry, gave a heuristic description of the behavior of high-energy wave-functions of quantum chaotic systems. He suggested that these should, in some sense, at the wavelength scale, behave like stationary Gaussian fields whose spectral measure is uniformly distributed on the unit sphere. The ambiguous comparison between a deterministic system and a stochastic field has given rise to many different interpretations. In the present paper, we are interested in a formulation given by one of the authors in \cite{LWL} (see also \cite{ABLM} for a similar approach). In this interpretation, we consider a compact {\color{black}connected} Riemannian manifold $(X^d,g)$ with negative sectional curvature. We will denote by $\mathrm{d}x$ the volume measure on $X$ and we will denote by $\Delta$ the Laplace-Beltrami operator on $X$. The conjecture can be roughly stated as follows: Let $(\psi_h)_h$ be a family of functions on $X$ such that $h^2\Delta\psi_h+\psi_h=0$ and normalized so that $\|\psi_h\|_2^2=1$. Let $\mathcal{U}\subset X$ be an open subset on which there exists a family of vector fields $(V_1,\dots,V_d)$ forming an orthonormal frame of the tangent bundle. Given $x\in \mathcal{U}$, we write $\widetilde{\exp}_x(y):=\exp_x(\sum_{j=1}^d y_jV_j(x))$. Let $\textsc{x}$ be a random point in $\mathcal{U}$ chosen uniformly with respect to the volume measure $\mathrm{d}x$. For each $h>0$ in the index set of $(\psi_h)_h$, let $\varphi_{\textsc{x}}^h\in C^\infty(\R^d)$ be the random field defined by $\varphi_{\textsc{x}}^h(y)=\psi_h(\widetilde{\exp}_{\textsc{x}}(hy))$. Then, the conjecture can be stated as follows:

\begin{center}
\textbf{Conjecture}: \textit{As $h\rightarrow 0$ in the index set of $(\psi_h)_h$, the family $\varphi_{\textsc{x}}^h(y)$ converges in law as a random field towards a stationary Gaussian field on $\R^d$ whose spectral measure is the uniform measure on the unit sphere $S^{d-1}$.}
\end{center}

This conjecture has many consequences in terms of nodal domains and semi-classical limits of $(\psi_h)$, as explained in \cite{LWL}. However, as stated, it seems quite out of reach.

\paragraph{Lagrangian states}

In this paper, instead, we study a much simpler question, in which eigenfunctions are replaced by a well-behaved family of quasi-modes, namely Lagrangian states:

\begin{tcolorbox}[breakable, enhanced]
\begin{definition}[Lagrangian states]\label{def:LagStat}
A \emph{Lagrangian state} is a family of functions $(u(\cdot;h))_h$ on $X$ indexed by $h\in]0,1[$, defined as follows
\begin{equation}\label{eq:LagState}
f_h(x)=a(x)e^{i\phi(x)/h}
\end{equation}
where $\phi\in C^\infty(U)$ for some open subset $U\subset X$ and $a\in C^\infty_c(U)$.  The \emph{energy measure} of $f_h$ is the measure on $(0,\infty)$, denoted by $\boldsymbol{\mu}_{a,\phi}$, which is the push-forward of the measure $|a(x)|^2 \mathrm{d}x$ on $X$ by the map $X\ni x \mapsto |{\color{black}\partial\phi(x)}|\in (0,\infty)$. We say that the Lagrangian state is \emph{monochromatic} if it furthermore satisfies $|\partial  \phi(x)|=1$ for all $x\in U$. In particular, this implies that $\boldsymbol{\mu}_{a,\phi}$ is a multiple of $\delta_{\{1\}}$.
\end{definition}
\end{tcolorbox}

Monochromatic Lagrangian states are quasimodes in the sense that they satisfy\footnote{\textcolor{black}{Here and in all the sequel, $O_{C^k}(h^\alpha)$ denotes a family of functions $(g_h)$ such that $\|g_h\|_{C^k(X)}$ is bounded by a constant times $h^\alpha$.}}
\[
h^2\Delta f_h(x)+f_h(x)=O_{C^0}(h)\, .
\]

However, the conjecture above will clearly not hold for them since they vanish on some non-empty open subset of $X$. 

Hence, instead of studying Lagrangian states of the form (\ref{eq:LagState}), we will study their evolution by the Schrödinger equation. It can be explicitely described using the WKB method, and is closely related to the dynamics of the geodesic flow.

Such a strategy was already followed in \cite{Schu}, where it was shown that a wide family of monochromatic Lagrangian states evolved during a long time have the Liouville measure as their semi-classical measure. Hence, they satisfy an analogue of quantum unique ergodicity, which is a central conjecture in quantum chaos concerning the genuine eigenfunctions of the Laplacian. In \cite{Schu}, the semi-classical measure associated to the long time evolution of non-monochromatic Lagrangian states is also described explicitly, as a linear combination of Liouville measures at different energies.

A precise description of the long time propagation of Lagrangian states was also used, for instance in \cite{Anan}, \cite{AN} and \cite{NZ}, to prove properties about the eigenfunctions and resonances of quantum chaotic systems.

It is thus natural to conjecture that (generic) Lagrangian states evolved during a long time satisfy the same quantum chaotic conjectures as genuine eigenfunctions of the Laplacian. In particular, we can wonder if they satisfy an analogue of Berry's conjecture stated above.

\paragraph{Informal presentation of our results}
The present paper gives a (partial) positive answer to this question. Namely, we consider a ``generic'' Lagrangian state, and propagate it to a time $t$ by the Schrödinger equation, which gives us a function $f_h^t$. {\color{black}To be more precise, recall first that a subset of a topological space is called \emph{residual} if it contains a countable} {\color{black}intersection of dense} {\color{black}open subsets. We will first equip the space of phases $\phi$ defined on the support of a fixed amplitude $a$ with a natural topology. We then construct a residual subset of the space of phases such that our result will hold under the condition that $f_h(x)=a(x)e^{i\phi(x)/h}$ where $\phi$ belongs to this subset.} Similarly to the construction in the previous paragraph, we write $f_{h,\mathrm{x}}^t(y):= f^t_h(\widetilde{\exp}_{\textsc{x}}(hy))$, with $\mathrm{x}$ chosen uniformly at random in some open set of $X$. We can then show that $f_{h,\mathrm{x}}^t$ admits a weak limit for all $t$ large enough, and that, as $t\longrightarrow +\infty$, this limit converges to an isotropic Gaussian field. In the special case where the initial state is monochromatic, we thus obtain the same limit as in Berry's conjecture.

There are two major differences between our results and those of \cite{Schu}. 
\begin{itemize}
\item In \cite{Schu}, the condition on Lagrangian states is completely explicit: one has to assume that the associated Lagrangian manifold is transverse to the stable directions of the classical dynamics (see section \ref{ss:Hyp} for more details). Here, we also need transversality to the stable directions, but also some much more subtle conditions. Namely, we will use the WKB method to express the evolved Lagrangian state, locally, as a sum of plane waves. We will need the fact that, generically, these plane waves have directions of propagation which are rationally independent, so that, when observing this sum of waves at a random point, it will behave like a sum of independent complex numbers with uniform argument. Gaussianity will then emerge from the central limit theorem.

\item In \cite{Schu}, the Lagrangian states are propagated up to the Ehrenfest time, that is, $c |\log h|$ for some $c>0$ related to the classical dynamics. Here, we first take $h$ to zero to define our limits, and then let $t$ go to infinity, which is somehow much weaker. We believe that an adaptation of our method could allow us to show Berry's conjecture for generic Lagrangian states propagated up to some time $c \log |\log h|$ for some $c>0$. However, to do so, we would have to change our definition of genericity, from $\phi$ belongs to a residual set here to $\phi$ belongs to a space of full measure for some suitable measure. This should be pursued elsewhere.
\end{itemize}

Despite these weaknesses, our result can be considered as the first example of a family of functions satisfying Berry's conjecture because of an underlying chaotic classical dynamics. Before that, \cite{Bou} and \cite{BW} (see also \cite{LWL} and \cite{Sar}) proved Berry's conjecture for generic families of Laplace eigenfunctions on the two dimensional torus, using some arithmetic arguments. Some examples of families of eigenfunctions in $\R^d$ satisfying Berry's conjecture are also given in \cite{RomSar}.

\paragraph{Organization of the paper}
In section \ref{sec:main_results}, we will present our main result, recalling all the definitions we need regarding local weak limits and Gaussian fields. In section \ref{sec:MainProof}, we will show that our main result holds, provided our initial state is a Lagrangian state whose phase $\phi$ belongs to a special set. We show in section \ref{sec:PhaseGen} that this set is in some sense generic. A key step in the proof of the main results presented in section \ref{sec:MainProof} is to give an explicit description of the action of the Schr\"odinger operator on Lagrangian states. This is Proposition \ref{prop:SumLag3}. The proof of this proposition is the object of section \ref{sec:ClasDyn}, where we recall some properties of the geodesic flow in negative curvature.
Finally, in Appendix \ref{sec:semi-classique}, we will recall the facts we need from semi-classical analysis, while in Appendix \ref{sec:AppMono}, we give a description of the monochromatic phases we consider.

\paragraph{Acknowledgements}
This project was partially supported by a CNRS grant \emph{Projet Exploratoire de Premier Soutien (PEPS) `` Jeune chercheuse, jeune chercheur ''}. 

{\color{black} We would like to thank the anonymous referees for their multiple remarks, which greatly helped improve the general presentation of the paper.}

\section{Set-up and main results}\label{sec:main_results}
Our main results state that Lagrangian states converge to some Gaussian field. Hence, we first have to explain our notion of convergence, and then, to describe the Gaussian fields towards which they converge. We will also need our Lagrangian states to be associated with Lagrangian manifolds that are transverse to the stable directions of the geodesic flow, as we will explain in section \ref{ss:Hyp}.

Recall that $(X,g)$ is a compact {\color{black}connected} Riemannian manifold. {\color{black} For each $x\in X$, we will denote by $\exp_x:T_xX\rightarrow X$ the exponential map at $x$ induced by the metric $g$ on $X$} {\color{black}(as in \cite[Definition 1.4.3]{Jost})}. {\color{black}Moreover, given $x,y\in X$, we will denote by $\mathrm{d}(x,y)$ the Riemannian distance between the two points $x$ and $y$.} Unless otherwise stated, the spaces $C^\infty(X)$ and $C^\infty(\R^d)$ will be equipped with the topology of uniform convergence of derivatives on compact sets. Moreover, when we speak of probability measures on these spaces, we will assume that they are equipped with the Borel $\sigma$-algebra.

\subsection{Local limits}\label{ss:local_limits_def}

Let us now describe the form of convergence we establish here. To avoid any topological difficulties, we define this convergence locally, though all of our results will hold regardless of the choice of localization. To the point, let $\mathcal{U}\subset X$ be a small enough open set so that we can define an orthonormal frame $V=(V_1,\dots,V_d)$ on it, that is to say a family of smooth sections $(V_i)_{i=1,\dots,d} : \mathcal{U}\longrightarrow TX$ such that, for each $x\in X$, $(V_1(x),\dots,V_d(x))$ is an orthonormal basis of $T_xX$.

If $x\in \mathcal{U}$ and $y\in \R^d$, we will write $yV(x) := y_1 V_1(x)+\dots+y_d V_d(x) \in T_xX$, and
\begin{equation}\label{eq:JamesBrownIsTheBest}
\widetilde{\exp}_{x}(y) := \exp_x (yV(x))\, .
\end{equation}

All the constructions in this section will depend on the choice of this local frame, and will hence not be intrinsic.  For the rest of the section, let us fix $\textsc{x}$ a random point in $\calU$ chosen uniformly with respect to the Riemannian volume measure.
\begin{tcolorbox}[breakable, enhanced]
\begin{definition}
Let $(f_h)_{h>0}$ be a family of functions in $C^\infty(X)$, and let $\P$ be a probability measure on $C^\infty(\R^d)$. Then, for each $h>0$, we define the \textit{$h$-local measure} associated to this family as the law of the random element of $C^\infty(\R^d)$ defined by $f_{\textsc{x},h}(y):=f_h(\widetilde{\exp}_{\textsc{x}}(hy))$. We say that $\P$ is the local weak limit of $(f_h)_h$ in the frame $V$ if, as $h\rightarrow 0$, the law of $f_{\textsc{x},h}$ converges weakly to $\P$.
\end{definition}
\end{tcolorbox}

We insist that, in the definition of $f_{\mathrm{x},h}$, $\mathrm{x}$ is a point chosen uniformly at random in $\mathcal{U}$, so that $f_{\mathrm{x},h}$ is a random element of $C^{\infty}(\R^d)$. Here, $C^\infty(\R^d)$ is equipped with its usual topology, given by uniform convergence of derivatives over compact sets.

Hence, saying that $\P$ is the local weak limit of $(f_h)_h$ in the frame $V$ means that, for any continuous bounded functional $F: C^\infty(\R^d) \longrightarrow \R$, we have
$$\frac{1}{\Vol(\mathcal{U})} \int_\mathcal{U} F (f_{\mathrm{x},h})\mathrm{dx} \underset{h \to 0}{\longrightarrow} \mathbb{E}_\P [F].$$

\begin{tcolorbox}[breakable, enhanced]
\begin{definition}
Let $(f_h)_{h>0}$ be a family of functions in $C^\infty(X)$, let $(r_h)_{h>0}$ be a family of positive real numbers converging to $0$, let $x_0\in\calU$ and let $\P_{x_0}$ be a probability measure on $C^\infty(\R^d)$. We say that $\P_{x_0}$ is the $(r_h)_h$-local limit of $(f_h)_{h>0}$ at $x_0$ (in the frame $V$) if, as $h\rightarrow 0$, the law of the random function $f_{\textsc{x},h}$, conditioned on the event that $\textsc{x}\in B(x_0,r_h)$, converges weakly to $\P_{x_0}$.
\end{definition}
\end{tcolorbox}

In other words, $\P_{x_0}$ is the $(r_h)_h$-local limit of $(f_h)_{h>0}$ at $x_0$ (in the frame $V$) if for any continuous bounded functional $F: C^\infty(\R^d) \longrightarrow \R$, we have
$$\frac{1}{\Vol(B(x_0, r_h))} \int_{B(x_0,r_h)} F (f_{\mathrm{x},h})\mathrm{d}x \underset{h \to 0}{\longrightarrow} \mathbb{E}_{\P_{x_0}} [F].$$

\begin{remark}\label{rem:PointwiseToGlobal}
By construction, if $(f_h)_h$ has an $(r_h)$-local limit $\P_{x_0}$ at almost every $x_0\in\calU$, then it has an $h$-local limit $\P$  which satisfies
\[
\P = \int_{\calU} \P_{x_0} \mathrm{d}x_0\, .
\]
\end{remark}

\subsection{Gaussian fields}\label{subsec:GaussField}
%{\color{black}Let us equip $C^\infty(X)$ with the Borel} {\color{black}$\sigma$-}{\color{black}algebra associated with the topology of uniform convergence of successive derivatives {\color{black}{on compact sets\footnote{Pourquoi parler de $X$ alors que tout sera sur $\R^d$ ?}}}. An almost surely (or a.s.) $C^\infty$ (centered) Gaussian field on $X$ will be a random variable $f$ taking values in $C^\infty(X)$ such that for any finite collection of points $x_1,\dots,x_k\in X$, the random vector $(f(x_1),\dots,f(x_k))\in \C^d$ is (centered) Gaussian. We say that two fields $f_1$ and $f_2$ are \emph{equivalent} if they have the same law. In the sequel, unless otherwise stated, {\color{black}{we will always}} identify fields which are equivalent. That is to say that we will speak indifferently of the field and of its law. Of course, this definition extends immediately to any smooth manifold and in particular $\R^d$ where $C^\infty(\R^d)$ is equipped with the topology of uniform convergence of derivatives over compact sets.\\

{\color{black} As previously, we equip $C^\infty(\R^d)$ with its usual topology, given by uniform convergence of derivatives over compact sets. An almost surely (or a.s.) $C^\infty$ (centered) Gaussian field on $\R^d$ will be a random variable $f$ taking values in $C^\infty(\R^d)$ such that for any finite collection of points $x_1,\dots,x_k\in \R^d$, the random vector $(f(x_1),\dots,f(x_k))\in \C^d$ is (centered) Gaussian. We say that two fields $f_1$ and $f_2$ are \emph{equivalent} if they have the same law. In the sequel, unless otherwise stated, {\color{black}{we will always}} identify fields which are equivalent. That is to say that we will speak indifferently of the field and of its law.}\\

{\color{black}Let $f$ be an a.s. $C^\infty$, centered Gaussian field on $\R^d$. Then, the {\color{black}covariance} function $K:(x,y)\mapsto E[f(x)\overline{f(y)}]$ defined on $\R^d\times \R^d$ is \emph{positive definite}, meaning that for each $k$-uple $(x_1,\dots,x_k)\in (\R^d)^k$, the matrix $K(x_i,x_j)_{i,j}$ is Hermitian. As explained for instance in Appendix A.11 of \cite{NS}, the function $K$ belongs to $C^\infty(\R^d\times \R^d)$ and there is actually a bijection between such functions and a.s. $C^\infty$ centered Gaussian fields on $\R^d$ (up to equivalence).\\

Next, recall that, by Bochner's theorem (see for instance \cite{CI}, section 2.1.11), given a finite Borel complex measure $\boldsymbol{\mu}$ on $\R^d$, the Fourier transform $\hat{\boldsymbol{\mu}}$
 of $\boldsymbol{\mu}$ gives rise a continuous positive definite function $K:(x,y)\mapsto \hat{\boldsymbol{\mu}}(x-y)$ on $\R^d\times\R^d$. If, in addition, $\boldsymbol{\mu}$ is compactly supported, its Fourier transform is smooth and gives rise to a unique Gaussian field $f$ on $\R^d$ (up to equivalence). In this case, we say that $\boldsymbol{\mu}$ is the \textit{spectral measure} of $f$. Note that $K$ is invariant by the diagonal action of translations on each of its variables. Consequently, the law of $f$ is invariant by translations. We say in this case that $f$ is \textit{stationary}.\\

Let us now apply this recipe to define a family of Gaussian fields on $\R^d$. Fix $0<\lambda_1<\lambda_2$, and let $\boldsymbol{\mu}$ be a Borel measure on $[\lambda_1,\lambda_2]$. Consider the measure $\boldsymbol{\lambda_\mu}$ on $\mathbb{R}^d$ which is given by
\begin{equation}\label{eq:measure_def}
\int_{\R^d} {\color{black}g}(x) \mathrm{d}\boldsymbol{\lambda_\mu}(x) = \int_{\lambda_1}^{\lambda_2} \int_{\mathbb{S}^{d-1}} {\color{black}g}(r y) \mathrm{d}\omega_{d-1}(y)   \mathrm{d} \boldsymbol{\mu}(r),
\end{equation}
where $\omega_{d-1}$ is the uniform measure on $\mathbb{S}^{d-1}$. If $\boldsymbol{\mu}= \boldsymbol{\mu}_{a,\phi}$ with $a, \phi$ as
 in Definition \ref{def:LagStat}, we simply write $\boldsymbol{\lambda}_{a,\phi}$ instead of $\boldsymbol{\lambda}_{\boldsymbol{\mu}_{a,\phi}}$.

\begin{tcolorbox}[breakable, enhanced]
\begin{definition}\label{def:spm_to_field}
The isotropic Gaussian field with energy decomposition $\boldsymbol{\mu}$ is the unique law for an a.s. continuous stationary Gaussian field $f$ on $\R^d$ whose spectral measure is $\boldsymbol{\lambda_\mu}$. In other words, for each $x,y\in\R^d$,
\[
\E[f(x)f(y)]=\int_{\R^d}e^{i(x-y)\cdot\xi}d\boldsymbol{\lambda_\mu}(\xi)\, .
\]
{\color{black}We will denote by} $\mathbb{P}_{\boldsymbol{\mu}}$ the law of $f$, which is a probability measure on $C^\infty(\R^d)$. If $\boldsymbol{\mu}=\delta_{\{1\}}$ we call $f$ the \textit{random monochromatic wave}.
\end{definition}
\end{tcolorbox}

Note that, when $|\partial \phi(x)|=1$ for all $x$ in the domain of $\phi$, $\boldsymbol{\mu}_{a,\phi}$ is $\|a\|_{L^2(X)}^2\delta_{\{1\}}$. In particular, if $f$ {\color{black} is a Gaussian field with law} $\mathbb{P}_{\boldsymbol{\mu}_{a,\phi}}$, then $\|a\|_{L^2(X)}^{-1} f$ is (equivalent to) the random monochromatic wave.
}

\subsection{Transversality to the stable directions}\label{ss:Hyp}

We denote by $\Phi^t:T^*X\rightarrow T^*X$, $t\in\R$ the geodesic flow on $T^*X$. For each $\lambda > 0$, let us write $S_\lambda^*X := \{ (x,\xi)\in T^*X ~| ~|\xi| = \lambda\}$. If $0<\lambda_1<\lambda_2$, we also write $S^*_{[\lambda_1,\lambda_2]}X := \bigcup_{\lambda \in [\lambda_1, \lambda_2]} S_\lambda^*X$. Since $X$ has negative curvature, $(\Phi^t)_t$, restricted to some $S_\lambda^*X$, is an Anosov flow (see \cite{Ebe} for a proof of this fact). We will recall in section \ref{ss:back_to_hyperbolicity} the definition of an Anosov flow. In particular, we defer to this section for the definition, for each  $\rho\in S_\lambda^*X$, of the \emph{unstable}, \emph{stable} and \emph{neutral} subspaces of $T_\rho S_\lambda^*X$. For any $0<\lambda_1<\lambda_2$  and any open subset $\Omega\subset X$, we write 
\begin{equation}\label{eq:SpacePhases}
\begin{aligned}
\mathcal{E}_{(\lambda_1,\lambda_2)}(\Omega)&:= \{ \phi\in C^\infty(\Omega) ~ \text{ such that } \lambda_1< |\partial \phi | < \lambda_2 \}.
\end{aligned}
\end{equation}
To each $\phi\in \mathcal{E}_{(\lambda_1,\lambda_2)}(\Omega)$, we can associate a Lagrangian manifold 
\[
\Lambda_\phi=\{(x,\partial\phi(x))\, :\, x\in\Omega\}\subset T^*X\, .
\]
We then define the set of phases associated to Lagrangian manifolds that are transverse to the stable directions as
\begin{equation}\label{eq:TransversePhases}
\mathcal{E}^{T}_{(\lambda_1,\lambda_2)}(\Omega) :=\big{\{} \phi \in \mathcal{E}_{(\lambda_1,\lambda_2)}(\Omega) \text{ such that } \forall x\in \Omega, T_{(x,\partial \phi(x))} \Lambda_\phi \cap E^-_{(x,\partial \phi(x))} = \{0\} \big{\}}.
\end{equation}

For each $\rho=(x,\xi)\in T^*X$ such that $\xi\neq 0$, let $\hat{E}^0_\rho=\{(0,s\xi)\, :\, s\in\R\}$ (as in section \ref{ss:back_to_hyperbolicity}). Then $E^+_\rho\oplus E^-_\rho\oplus E^0_\rho \oplus \hat{E}^0_\rho= T_\rho T^*X$. Hence, $\phi\in\calE^T_{(\lambda_1,\lambda_2)}(\Omega)$ if and only if
\[
T_{(x,\partial\phi(x))}\Lambda_\phi\oplus E^-_{(x,\partial\phi(x))}\oplus \hat{E}^0_{(x,\partial\phi(x))}=T_{(x,\partial\phi(x))}T^*X\, .
\]

\subsection{Convergence of Lagrangian states to Gaussian fields}
Our main result does not hold for all Lagrangian states, but only for a generic subset of  $\mathcal{E}^{T}_{(\lambda_1,\lambda_2)}(\Omega)$, which we equip with the topology of uniform convergence of derivative on compact sets.

\begin{remark}\label{rem:NonEmpty}
The set $\mathcal{E}^T_{(\lambda_1,\lambda_2)}(\Omega)$ is not open, but if $\phi \in  \mathcal{E}_{(\lambda_1,\lambda_2)}^T(\Omega)$ and if $\overline{\Omega'} \subset \Omega$, then $ \mathcal{E}_{(\lambda_1,\lambda_2)}^T(\Omega')$ contains a neighbourhood of $\phi_{|\Omega'}$. This follows directly from the fact that $\rho\mapsto E_\rho^-$ is continuous. Furthermore, if $(x,\xi)\in S_{[\lambda_1,\lambda_2]}^*X$, we know that $E_{(x,\xi)}^+\cap E_{(x,\xi)}^-= \{0\}$. 

Therefore, if $\phi \in \mathcal{E}_{(\lambda_1,\lambda_2)}(X)$ and $x_0\in \Omega$ are such that the image of $\mathrm{d}_{x_0} (x, \partial \phi(x))$ is included in $E_{(x_0, \partial \phi(x_0))}^+$, then if $\Omega'$ is a small enough neighbourhood of $x_0$, we will have $\phi'_{|\Omega'} \in \mathcal{E}^T_{(\lambda_1,\lambda_2)}(\Omega')$ for any $\phi'\in C^\infty(\Omega)$ close enough to $\phi$. Therefore, if $\Omega'$ is small enough, $\mathcal{E}^T_{(\lambda_1,\lambda_2)}(\Omega')$ is non-empty, and contains $\phi'|_{\Omega'}$ for any $\phi'$ in a non-empty open subset of $C^\infty(\Omega)$.
 \end{remark}

We may now state our main result. To this end, we introduce the semi-classical Schr\"odinger propagator $U_h(t):=e^{ith\frac{\Delta}{2}}:L^2(X)\rightarrow L^2(X)$. {\color{black}Moreover, we recall once more that, a subset of a topological space is called residual if it contains a countable intersection of dense open subsets.}

\begin{tcolorbox}[breakable, enhanced]
\begin{theorem}\label{th:Pointwise}
Let $X$ be a compact {\color{black}connected} Riemannian manifold with negative sectional curvature, let $0<\lambda_1<\lambda_2$, and let $\Omega\subset X$ be an open subset. Then there exists a residual subset $\mathcal{E}^{T,irr}_{(\lambda_1,\lambda_2)}(\Omega)$ of $\mathcal{E}^T_{(\lambda_1,\lambda_2)}(\Omega)$ such that for any $\phi\in  \mathcal{E}^{T,irr}_{(\lambda_1,\lambda_2)}(\Omega)$, there exists $T_0\geq 0$ such that the following holds. Let $a\in C^\infty_c(\Omega)$. For each $h\in]0,1[$, we write $f_h(x)=a(x)e^{i\phi(x)/h}$.\\

Let $\mathcal{U}\subset X$ be an open set, and $V$ be an orthonormal frame on $\mathcal{U}$.   Let $\frac{1}{2}<\alpha<1$. Then for almost every $x_0\in \mathcal{U}$, and every $t\geq T_0$, the family $\br{U_h(t) f_h}_{h>0}$ has an $(h^\alpha)_{h>0}$-pointwise local weak limit at $x_0$, which we denote by $\mu_{t,x_0}$. Furthermore, $\mu_{t,x_0}$ converges weakly to $\mathbb{P}_{\boldsymbol{\mu}_{a, \phi}}$ as $t\to +\infty$.
\end{theorem}
\end{tcolorbox}

Thanks to Remark \ref{rem:PointwiseToGlobal}, Theorem \ref{th:Pointwise} implies the following result.

\begin{tcolorbox}[breakable, enhanced]
\begin{corollary}
With the notation of Theorem \ref{th:Pointwise}, for each $t\geq T_0$, the family $\left(U_h(t) f_h\right)_{h>0}$ has an $h$-local limit $\mu_t$ which converges weakly to {\color{black} $\mathbb{P}_{\boldsymbol{\mu}_{a, \phi}}$} as $t\to+\infty$.
\end{corollary}
\end{tcolorbox}
{\color{black}
\begin{remark}
Note that, although the law $f_h$ depends on $\calU$ and on the choice of frame $V$, the limiting measure $\mathbb{P}_{\boldsymbol{\mu}_{a, \phi}}$ depends only on $a$ and $\phi$.
\end{remark}
}
\begin{remark}
Let us finally observe that by Remark \ref{rem:NonEmpty}, each point of $X$ admits an open neighbourhood $\Omega\subset X$ for which $\calE^{T,irr}_{(\lambda_1,\lambda_2)}(\Omega)$ is non empty (and even uncountable). Hence, although we do not have a global ``generic" statement, Theorem \ref{th:Pointwise} does yield a wide family of Lagrangian states whose pointwise local weak limits converge to that of the isotropic stationary a.s. smooth Gaussian field on $\R^d$ with spectral measure $\boldsymbol{\mu_{a,\phi}}$ from Definition \ref{def:LagStat} as $t\to+\infty$, under the action of the Schr\"odinger flow.
\end{remark}

\subsection{The case of monochromatic phases}\label{sec:IntroMono}

We would now like to state an analogue of Theorem \ref{th:Pointwise} for monochromatic phases, i.e., phases satisfying{\footnote{The case $|\partial \phi|=\lambda$ for some $\lambda>0$ can be recovered from the case $|\partial  \phi|=1$ by a simple rescaling.}  $|\partial  \phi|=1$. At first glance, it would seem natural to work with the space of phases
$$\mathcal{E}_{1}(\Omega) = \{ \phi\in C^\infty(\Omega) ~ \text{ such that } |\partial \phi |=1 \},$$ which we would equip with the $C^\infty(\Omega)$ topology. However, this set appears to be very hard to work with: it is not trivial to perturb a function in $\mathcal{E}_{1}(\Omega)$  while remaining in this set. Hence, the set $\mathcal{E}_{1}(\Omega)$ could contain isolated points, which would make our approach based on genericity irrelevant. We will therefore use another approach to study phases satisfying $|\partial  \phi|=1$.\\

Let $\Sigma\subset X$ be an embedded orientable simply connected hypersurface. Let us denote by $\nu$ a vector field defined on $\Sigma$ such that for each $y\in\Sigma$, $\nu(y)$ has unit norm and is orthogonal to $T_y\Sigma$. We write
\begin{equation}\label{eq:calC_def}
\calC(\Sigma)=\{u\in C^\infty(\Sigma)\, :\, |\partial  u|<1\}\, .
\end{equation}
If $u\in \mathcal{C}(\Sigma)$, we define, for any $y\in \Sigma$, $v_u(y):=\partial _yu+(1-|\partial _yu|^2)^{1/2}\nu(y)\in S^*_y X$, and
\[
L_u:= \{ (y, v_u(y))\, :\, y\in \Sigma\}\, .
\]
We then define
\begin{equation}\label{eq:calC_trans}
\mathcal{C}^T(\Sigma) := \left\{u \in \mathcal{C}(\Sigma)\, \big|\,  \forall y \in \Sigma, T_{(y, v_u(y))}L_u \cap \left( E^+_{(y, v(y))} \oplus E^0_{(y, v_u(y))}\right) = \{0\} \right\}.
\end{equation}

By Lemma \ref{l:existence_uniqueness}, given $u\in \mathcal{C}(\Sigma)$, there exists an open neighbourhood $\Omega_u\subset X$ of $\Sigma$, and a map $\phi_u\in \mathcal{E}_1(\Omega_u)$ such that
\begin{equation}\label{eq:phi_from_u}
\phi_u|_\Sigma=u\text{ and }\partial\phi_u|_\Sigma=v_u\, .
\end{equation}
Moreover, any two functions with these properties must coincide on a neighbourhood of $\Sigma$. Furthermore, by Lemma \ref{l:description_eikonal}, for any $x\in \Omega_u$, {\color{black}there exists a unique pair $(y,t)\in\Sigma \times \R$ such that}
\begin{equation}\label{eq:GradientsAreCool}
(x, \partial \phi_u(x)) = \Phi^t(y, v_u(y))\, .
\end{equation}

\begin{figure}
\center
\includegraphics[scale=0.4]{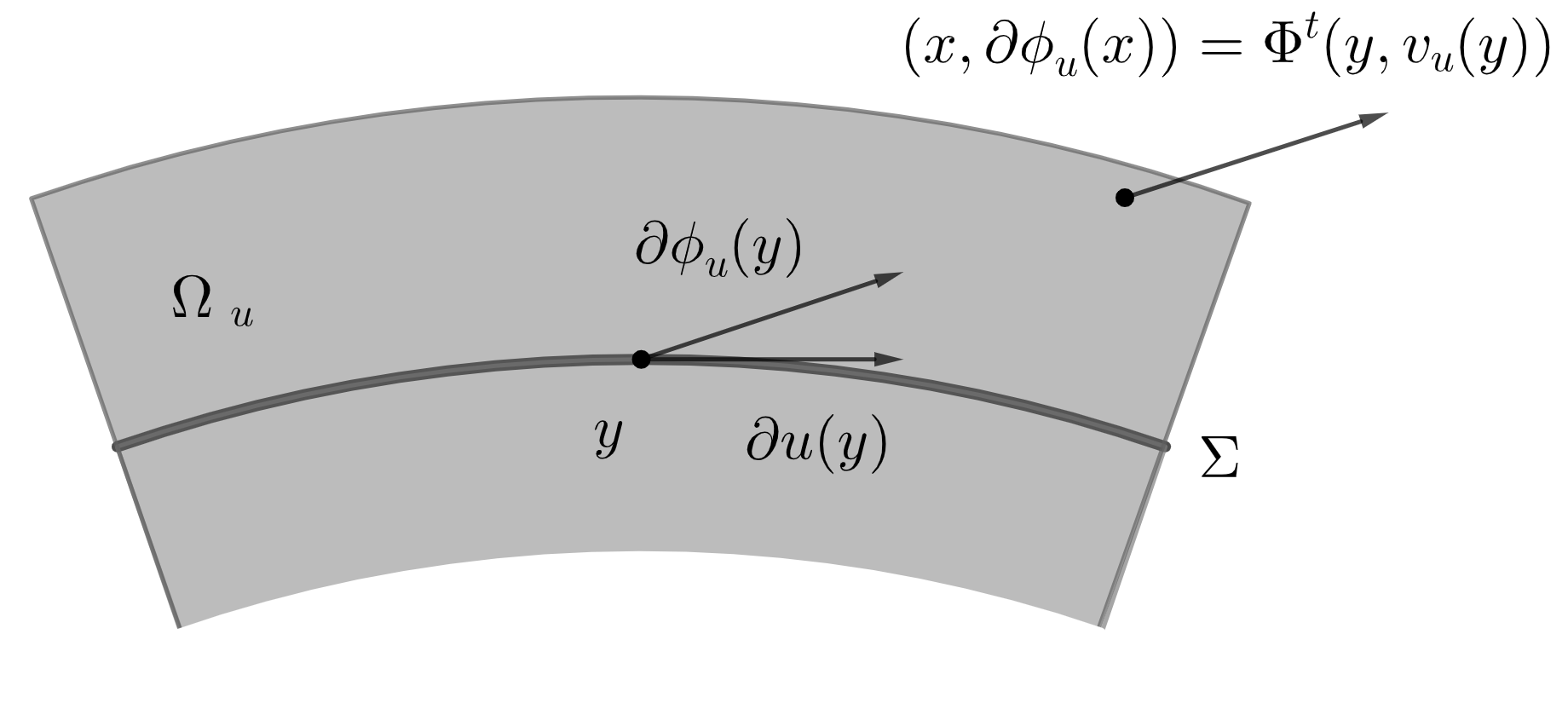}
            \caption{Construction of the phase $\phi_u$ from $u$.}
            \label{fig:Zden}
\end{figure}

In particular, we see from (\ref{eq:DirInv}) that $u\in \mathcal{C}^T(\Sigma)$ if and only if {\color{black} there exists $\Omega_u'\subset X$ an open subset with $\Sigma\subset\Omega_u'\subset\Omega_u$ such that $\phi_u \in \mathcal{E}_{1}^T(\Omega_u')$}. The same argument as in Remark \ref{rem:NonEmpty} shows that $\calC^T(\Sigma)$ is non-empty when $\Sigma$ is small enough, and that, if $u\in \calC^T(\Sigma)$ and $\overline{\Sigma'}\subset \Sigma$, then $\calC^T(\Sigma')$ contains a neighbourhood of $u_{|\Sigma}$.\\

We may now state our analogue of Theorem \ref{th:Pointwise} for monochromatic phases. To this end, we equip the set $\mathcal{C}(\Sigma)$ with the $C^\infty(\Sigma)$ topology (i.e., the topology of uniform convergence of derivatives on compact sets). Note that, unlike in the polychromatic case, the pointwise local weak limits exist here for all $x_0$, and not just for almost all of them.

\begin{tcolorbox}[breakable, enhanced]
\begin{theorem}\label{theoMono}
Let $X$ be a compact {\color{black}connected} Riemannian manifold with negative sectional curvature, and let $\Sigma\subset X$ be an embedded orientable simply connected hypersurface with a normal vector field $\nu$. There exists a residual subset $\mathcal{C}^{T, irr}(\Sigma)$ of $\mathcal{C}^T(\Sigma)$ such that, for any $u\in \mathcal{C}^{T, irr}(\Sigma)$, there exists $T_0\geq 0$ such that the following holds. Let $\phi_u$ and $\Omega_u$ be as in \eqref{eq:phi_from_u} {\color{black} such that $\phi_u\in\calE^T_1(\Omega_u)$}, and let $a\in C^\infty_c(\Omega_u)$. For each $h\in]0,1[$, we write $f_h(x)=a(x)e^{i\phi(x)/h}$. Let $\mathcal{U}\subset X$ be an open set, and $V$ be an orthonormal frame on $\mathcal{U}$.   Let $\frac{1}{2}<\alpha<1$. Then for every $x_0\in \mathcal{U}$, and every $t\geq T_0$, the family $\br{U_h(t) f_h}_{h>0}$ has an $(h^\alpha)_{h>0}$-pointwise local weak limit at $x_0$, which we denote by $\mu_{t,x_0}$. Furthermore, $\mu_{t,x_0}$ converges weakly to $\mathbb{P}_{\boldsymbol{\mu}_{a, \phi}}$ as $t\to +\infty$.
\end{theorem}
\end{tcolorbox}}
{\color{black}
\begin{remark}
Note that, in this case, as explained in section \ref{sec:IntroMono}, if $f$ has law $\mathbb{P}_{\boldsymbol{\mu}_{a, \phi}}$, then, $\|a\|_{L^2}^{-1}f$ is in fact the monochromatic wave. In particular, although the construction depends on $\calU$, on the choice of frame $(V(x))_x$, on $a$ and on $\phi$, the limit is (up to a multiplicative constant) independent of all of these choices.
\end{remark}
}
\begin{remark}
As for the case of Theorem \ref{th:Pointwise}, $\calC^{T,irr}(\Sigma)$ is non-empty and we obtain a wide family of Lagrangian states have pointwise local weak limits converging to the monochromatic wave under the action of the Schr\"odinger flow.
\end{remark}

\section{Proof of Theorems \ref{th:Pointwise} and \ref{theoMono}}\label{sec:MainProof}
The aim of this section is to describe explicitly the sets $\mathcal{E}^{T,irr}_{(\lambda_1,\lambda_2)}(\Omega)$ and $\calC^{T,irr}(\Sigma)$ appearing respectively in the statements of Theorem \ref{th:Pointwise} and Theorem \ref{theoMono}, and to prove these theorems, postponing the proof of the fact that $\mathcal{E}^{T,irr}_{(\lambda_1,\lambda_2)}(\Omega)$ (resp. $\calC^{T,irr}(\Sigma)$) is a residual subset of $\mathcal{E}^{T}_{(\lambda_1,\lambda_2)}(\Omega)$ (resp. $\calC^{T}(\Sigma)$)  to the next section.\\

Throughout the present section, we will therefore fix $\Omega\subset X$ an open subset, as well as constants $0<\lambda_1<\lambda_2$, and consider phases in $\calE_{(\lambda_1,\lambda_2)}(\Omega)$. Likewise, for the monochromatic case, we fix $\Sigma\subset X$ a simply connected embedded orientable hypersurfaces of $X$ and $\nu$ a section of $TX|_\Sigma$ such that for each $y\in\Sigma$, $\nu(y)$ has unit norm and is orthogonal to $T_y\Sigma$ in $T_yX$. We will also consider monochromatic phases of the form $\phi_u$ with $u\in\calC(\Sigma)$ as defined in section \ref{sec:IntroMono}.\\

Finally, in order to describe local limits, we also fix $\calU\subset X$ equipped an orthonormal frame $V$ as in section \ref{ss:local_limits_def}.\\

The proof will go as follows. In section \ref{ss:criterion} we state a compactness criterion. Thanks to this criterion, proving convergence of finite marginals will yield convergence in $C^\infty(\R^d)$ topology. In section \ref{ss:propagation} we will describe the effect of the Schr\"odinger propagator on a Lagrangian state whose phase belongs to $\calE^T_{(\lambda_1,\lambda_2)}(\Omega)$. In section \ref{ss:convergence} we first describe the sets $\calE^{T,irr}_{(\lambda_1,\lambda_2)}(\Omega)$ and $\calC^{T,irr}(\Sigma)$. Assuming that $\phi$ belongs to one of these sets we let $h\to 0$ for some fixed (large enough) $t$ and describe the local limits associated to the propagated Lagrangian state at time $t$ around some point $x_0$ (which we assume to be generic in the former case). In section \ref{ss:long_time} we let $t\to+\infty$ and describe the asymptotic behavior of the local limit around $x_0$. Finally, in section \ref{ss:conclusion} we fit the pieces together and complete the proofs of Theorems \ref{th:Pointwise} and \ref{theoMono}.

\subsection{A criterion for convergence of local measures}\label{ss:criterion}
Here we record a compactness criterion for the convergence of probability measures on $C^\infty(\R^d)$. Let $\boldsymbol{a} = (a_{k,\ell})_{k,\ell \in \N^2}$ be a sequence of positive real numbers depending on two parameters. We define
\begin{equation}\label{eq:CompSet}
\mathcal{K}(\boldsymbol{a}):= \{f\in C^\infty(\R^d) ~|~\forall k, \ell \in \N,  \|f\|_{C^\ell (B(0,k))} \leq a_{k,\ell}\}.
\end{equation}

It follows from the Arzela-Ascoli theorem that $\mathcal{K}(\boldsymbol{a})$ is a compact subset of $C^\infty(\R^d)$ for the topology of convergence of all derivatives over all compact sets.

Let us write $\mathcal{F}$ for the set of functionals $F$ of the form
 $\textcolor{black}{C^\infty(\R^d)\ni f\mapsto }F(f) = G(f(x_1),\dots, f(x_k)),$
where $k\in \N$, $x_1,\dots, x_k\in \R^d$ and $G\in C_c(\C^k)$.
Then $\mathcal{F}$ forms an algebra which separates points. Hence, by the Prokhorov theorem, we obtain the following result, which we will use several times in the sequel. See section 3 of \cite{LWL} for more details.
 
 \begin{tcolorbox}[breakable, enhanced]
\begin{lemma}\label{lem:CritMes}
Let $\boldsymbol{a} = (a_{k,\ell})_{k,\ell \in \N^2}$ be a sequence of positive real numbers depending on two parameters. Let $(\P_n)$ be a sequence of Borel probability measures on $C^\infty(\R^d)$, which is supported in $\mathcal{K}(\boldsymbol{a})$, and let $\mu$ be a Borel probability measure on $C^\infty(\R^d)$. Suppose that, for any $F\in \mathcal{F}$, we have
$$\mathbb{E}_{\mathbb{P}_n} \left[F\right] \underset{n\to +\infty}{\longrightarrow}  \mathbb{E}_{\mathbb{P}} \left[F\right].$$
Then $(\mathbb{P}_n)$ converges weakly to $\mathbb{P}$.
\end{lemma}
\end{tcolorbox}

\begin{remark}\label{rem:CondEsperance}
More generally, using Markov inequality, the condition that $(\P_n)$ is supported in $\mathcal{K}(\boldsymbol{a})$ can be replaced by the following:
For every $k, \ell\in \N$, there exists $a_{k,\ell}>0$ such that for all $n\in \N$, we have
$$\mathbb{E}_{\mathbb{P}_n} \left( \| f\|_{C^\ell (B(0,k))} \right)\leq a_{k,\ell}$$
\end{remark}

\subsection{Propagation of Lagrangian states by the Schrödinger equation}\label{ss:propagation}

In this subsection, we describe the propagation of Lagrangian states by the Schr\"odinger equation. In the classical world, each Lagrangian state $ae^{i\phi/h}$ defined on  $\Omega$ corresponds to a Lagrangian submanifold $\Lambda_\phi=\{(x,\partial\phi(x))\, :\, x\in\Omega\}\subset T^*X$. The dynamics of a Lagrangian state by the Schrödinger flow is easy to describe in terms of the evolution of $\Lambda_\phi$ under the geodesic flow on $X$. The main point of this section is to describe the effect of the Schrödinger propagator acting on a Lagrangian state on a manifold $X$ of negative sectional curvature. We do so in Proposition \ref{prop:SumLag3}. The proof of this proposition, which is essentially an application of the WKB method, relies on the techniques developed in \cite{Anan}, \cite{AN}, \cite{NZ}, and we will recall it in the section \ref{ss:demo_propagation} below for the reader's convenience. Recall that $U_h(t)=e^{ith\frac{\Delta}{2}}$ is the Schr\"odinger propagator and that $\Phi^t:T^*X\rightarrow T^*X$ is the geodesic flow. 

\begin{tcolorbox}[breakable, enhanced]
\begin{proposition}[Dynamics under the Schr\"odinger propagator]\label{prop:SumLag3}
Let $\phi_0\in \mathcal{E}^{T}_{(\lambda_1,\lambda_2)}(\Omega)$.

Then there exists $T_0=T_0(\phi_0)\geq 0$ such that for any $a\in C_c^\infty(\Omega)$ and any $t\geq T_0$, there exists $M(t)\in \N$ such that the application of the operator $U_h(t)$ to the  Lagrangian state
\begin{equation*}
a(x) e^{i\phi_0(x)/h}
\end{equation*}
 can be written, for any $k\in \N$, as
\begin{equation}\label{eq:PropagLagSum2}
U_h(t) (a e^{i\phi_0/h})(x) = \sum_{j=1}^{M(t)}  e^{i\phi_{j,t}(x)/h} b_{j,t}(x) +O_{C^k}(h),
\end{equation}
where  $b_{j,t}\in C^\infty(X)$ \textcolor{black}{are smooth functions whose support we denote by $\mathcal{U}_{j,t}$} and $\phi_{j,t}\in C^{\infty}(\mathcal{U}_{j,t})$, satisfying:
\begin{enumerate}
\item As $t\to +\infty$, $\max\limits_{j=1,\dots, N(t)} \|b_{j,t}\|_{C^0}\to 0$. Furthermore, for all $\varepsilon>0$, there exists $\delta>0$ such that for all $t\geq T_0$, all $j\in \{1,\dots, M(t)\}$ and all $x,y\in \textcolor{black}{\mathcal{U}_{j,t}}$, we have 
\begin{equation}\label{eq:LaBorneLog}
\mathrm{d}(x,y)\leq \delta \Longrightarrow |b_{j,t}(x)|\leq (1+\varepsilon) |b_{j,t}(y)|.
\end{equation}
\item For each $t\geq T_0$, $j \in\{1,\dots,M(t)\}$,  and $x\in \textcolor{black}{\mathcal{U}_{j,t}}$  there exists a point $y_{j,x,t}\in \Omega$ such that $\Phi^t(y_{j,x,t}, \partial \phi_0(y_{j,x,t})) = (x, \partial \phi_{j,t}(x))$. In particular,
$|\partial\phi_{j,t}|\in [\lambda_1,\lambda_2]$. 

\item There exists $C_1>0$ such that, for all $t\geq T_0$, all $j\in\{1,\dots,M(t)\}$ and all $x_0\in \mathcal{U}_{j,t}$, the number of $j'\in \{1,\dots,M(t)\}$ such that $x_0\in \mathcal{U}_{j',t}$ and $\partial _{x_0} \phi_{j,t} = \partial _{x_0} \phi_{j',t}$  is at most $C_1$.

\item There exists a constant $C_2>0$ such that for all $t\geq T_0$ and all $j\in\{1,\dots,M(t)\}$, we have
\begin{equation}\label{eq:phase_bound}
\|\partial\phi_{j,t}\|_{C^1}\leq C_2\ .
\end{equation}
\end{enumerate}
\end{proposition}
\end{tcolorbox}

For the rest of the section, we fix $a\in C^\infty_c(\Omega)$ and $\phi\in\calE_{(\lambda_1,\lambda_2)}^T$. For each $h>0$ and $t\in\R$, we set

\[
f^t_h:= U_h(t) (a e^{i\phi/h})\, .
\]
Proposition \ref{prop:SumLag3} applies to $f^t_h$. For each $x_0\in X$, $t\geq T_0$, $j,j'\in \{1,\dots,M(t)\}$, we will write $j\sim_{x_0,t} j'$ if $x_0\in \mathcal{U}_{j,t}\cap \mathcal{U}_{j',t}$ and  $\partial \left( \phi_{j,t}\circ \widetilde{\exp}_{x_0}  \right)(0)= \partial \left( \phi_{j',t}\circ \widetilde{\exp}_{x_0}  \right)(0)$. Up to reordering the terms $\{1,\dots,M(t)\}$, we may suppose that there exists $N(t;x_0)\in \N$ such that  the set $\{1,\dots,N(t;x_0)\}$ contains exactly one representative of each of the different equivalence classes. In the sequel, since $x_0$ will be fixed most of the time, we will just write $N(t)$ instead of $N(t,x_0)$.

We then write, for every $j\in \{1,\dots,N(t)\}$
{\color{black}
\begin{equation}\label{eq:xi_beta}
\begin{aligned}
\xi^{t,x_0}_j &:= \partial \left( \phi_{j,t}\circ \widetilde{\exp}_{x_0}  \right)(0)\in \mathbb{R}^{d};\, \qquad \boldsymbol{\xi^{t,x_0}}:=(\xi_1^{t,x_0},\dots,\xi_{N(t)}^{t,x_0}); \\
 B_{j,t} (x_0) &:=  \sum_{\underset{j'\sim_{x_0,t} j}{j'\in \{1,\dots,M(t)\}}} b_{j',t}(x_0) e^{i \phi_{j,t}(x_0)/h} \in \C\\
\beta^{t, x_0}_j&:= \left| B_{j,t} (x_0) \right| ;\, \qquad  \boldsymbol{\beta^{t,x_0}}:=(\beta_1^{t,x_0},\dots,\beta_{N(t)}^{t,x_0})\\
.
\end{aligned}
\end{equation}}

\subsection{Convergence to pointwise local limits at fixed times}\label{ss:convergence}

In this subsection, we first define the residual sets of phases \eqref{eq:DefIrr} and \eqref{eq:DefGeneMono} which appear in the statements of Theorems \ref{th:Pointwise} and \ref{theoMono} respectively. Then, assuming that the phase belongs to \eqref{eq:DefIrr} we describe the pointwise local limits at fixed time $t$ large enough (see Proposition \ref{prop:LWSumLag} below).\\

Recall the definitions of $\mathcal{E}_{(\lambda_1,\lambda_2)}(\Omega)$ \eqref{eq:SpacePhases} and $\mathcal{E}^{T}_{(\lambda_1,\lambda_2)}(\Omega)$ \eqref{eq:TransversePhases}.  Let us write 
\begin{equation}\label{eq:DefIrr}
\begin{aligned}
\mathcal{E}^{T, irr}_{(\lambda_1,\lambda_2)}(\Omega):=\big{\{}& \phi \in \mathcal{E}^T_{(\lambda_1,\lambda_2)}(\Omega)\, |\, \exists T_0(\phi)<+\infty \text{ such that for almost every $x_0\in X$,}\\
&\text{the vectors $(\xi_j^{t,x_0})_{j=1,...,N_{x_0}(t)}$ are rationally independent for all $t\geq T_0(\phi)$}\big{\}},
\end{aligned}
\end{equation}
where the $(\xi_j^{t,x_0})_j$ are obtained from $\phi$ by the construction \eqref{eq:xi_beta} which follows from Proposition \ref{prop:SumLag3}. The set $\mathcal{E}^{T, irr}_{(\lambda_1,\lambda_2)}(\Omega)$ is precisely the set appearing in the statement of Theorem \ref{th:Pointwise}. We will show in section \ref{sec:PhaseGen} that the space is a residual subset of $\mathcal{E}^{T}_{(\lambda_1,\lambda_2)}(\Omega)$ equipped with the convergence of all derivatives on all compact sets.\\

For the monochromatic case, we will consider the following analogous set. Recall that, in section \ref{sec:IntroMono}, given an oriented hypersurface $\Sigma$, we saw how to associate to each function $u\in \mathcal{C}(\Sigma)$ an open neighbourhood $\Omega_f$ of $\Sigma$ and a map $\phi_u \in \mathcal{E}_1(\Omega_u)$. If $u\in \mathcal{C}^T$, we thus denote by $(\xi_j^{t,x_0})_{j=1,\dots,N_{x_0}(t)}$ the vectors obtained by applying Proposition \ref{prop:SumLag3} to $\phi_u$ (see \eqref{eq:xi_beta}).
We then define 
\begin{align}\label{eq:DefGeneMono}
\mathcal{C}^{T,irr}(\Sigma) := \Big{\{} &f\in \mathcal{C}^T(\Sigma) \, |\, \exists T_0(\phi)<+\infty \text{ such that  for  every $x_0\in X$},\\ \nonumber
&\text{the vectors $(\xi_j^{t,x_0})_{j=1,\dots,N_{x_0}(t)}$ are rationally independent for all $t\geq T_0(\phi)$}\big{\}}.
\end{align}
We will see in section \ref{subsec:mono} that this set is a residual subset of $\mathcal{C}^{T}(\Sigma) $ equipped with the topology of uniform convergence of derivatives on compact sets.\\

From now on, we will always suppose that the phase $\phi$ introduced in section \ref{ss:propagation} belongs to $\mathcal{E}^{T,irr}_{(\lambda_1,\lambda_2)}(\Omega)$, and take $x_0$ such that the vectors $(\xi_j^{t,x_0})_{j=1,\dots,N_{x_0}(t)}$ are rationally independent for all $t\geq T_0(\phi)$.\\

Let us now describe the measures $\mathbb{P}_{t,x_0}$ appearing in Theorem \ref{th:Pointwise} associated to the family $(f_h^t)$ introduced in section \ref{ss:convergence}. To do this, recall that at the beginning of section \ref{sec:MainProof} we fixed $\calU$ an open subset of $X$ equipped with an orthonormal frame $V$. We will always implicitly consider $h$-local limits in this frame. The local limits of $(f_h^t)$ for various fixed $t$ will belong to a family of probability laws on $C^\infty(\R^d)$ which we now define:

\begin{tcolorbox}[breakable, enhanced]
\begin{definition}\label{def:rpws}
Let $N\in \N$,  $\boldsymbol{\beta}= (\beta_1,\dots,\beta_N)\in \left(\R^+\right)^N$, and $\boldsymbol{\xi} =(\xi_1,\dots, \xi_N) \in \left(\R^d\right)^N$. We associate to $(\boldsymbol{\beta}, \boldsymbol{\xi})$ a probability measure $\mathbb{P}_{\boldsymbol{\beta},\boldsymbol{\xi}}$ on $C^\infty(\R^d)$ as follows. Let $\vartheta_1,\dots,\vartheta_N$ be i.i.d uniform random variables in $[0,2\pi]$. Then, $\mathbb{P}_{\boldsymbol{\beta},\boldsymbol{\xi}}$ is the law of
\[
y\mapsto\sum_{j=1}^N\beta_je^{iy\cdot\xi_j+i\vartheta_j}\, .
\]
\end{definition}
\end{tcolorbox}

\begin{tcolorbox}[breakable, enhanced]
\begin{proposition}[Pointwise local limits in fixed time]\label{prop:LWSumLag}
Let $\frac{1}{2}<\alpha<1$, and let $t\geq T_0(\phi)$. Let $x_0\in\calU$ be such that the vectors $(\xi_j^{t,x_0})_{j=1,\dots,N_{x_0}(t)}$ are rationally independent. Then $(f_h^t)_h$ has an $h^\alpha$-pointwise local weak limit at $x_0$, which is given by $\mathbb{P}_{\boldsymbol{\beta}^{t,x_0}, \boldsymbol{\xi}^{t,x_0}}$.
\end{proposition}
\end{tcolorbox}

\begin{proof}[Proof of Proposition \ref{prop:LWSumLag}]
{\color{black}\textbf{First step: a criterion for convergence}}

Let $t\geq T_0(\phi)$ and $x_0\in\calU$ be such that the vectors $(\xi_j^{t,x_0})_{j=1,\dots,N_{x_0}(t)}$ are rationally independent. 
\textcolor{black}{Equation (\ref{eq:PropagLagSum2}) implies that, for any $R>0$ and any $k\in \N$, we have} 
$$\textcolor{black}{\|f_{\mathrm{x},h}^t\|_{C^k(B(0,R))}\leq  C(k) \sum_{j=1}^{M(t)} \|\phi_{j,t}\|_{C^k(B(0,R))} \|b_{j,t}\|_{C^k(B(0,R))}  +O(h).}$$
\textcolor{black}{This quantity is thus bounded independently of $h$, $t$ being fixed. This implies that} we may find a sequence $\boldsymbol{a}$ such that for all $h$ small enough and all $\mathrm{x}$ in $B(x_0,h^\alpha)$, the function $f^t_{\textsc{x},h}$ belongs to $\mathcal{K}(\boldsymbol{a})$, with $\mathcal{K}(\boldsymbol{a})$ as in (\ref{eq:CompSet}). Hence, thanks to Lemma \ref{lem:CritMes}, it suffices to show that for any $k\in \N$, any $y_1,\dots,y_n\in \R^d$ and any $G\in C_c(\C^n)$, we have 
\[
\mathbb{E} \left[F(f^t_{\textsc{x},h})\right] \underset{h\to 0}{\longrightarrow} \mathbb{E}_{\mathbb{P}_{\boldsymbol{\beta},\boldsymbol{\xi}}}\left[ F\right],
\]
where $F(f) = G(f(x_1), \dots, f(x_n))$ and where
the first {\color{black}expectation} is taken with respect to $\textsc{x}\in B(x_0, h^\alpha)$.\\

{\color{black}\textbf{Second step: Local expressions} 

Next, we are going to use Taylor expansions to obtain a simpler asymptotic expression for $F(f^t_{\textsc{x},h})$. If $x\in B(x_0, h^{\alpha})$, the fact that $b_{j,t}$ is $C^1$ implies that, for every fixed $y \in \R^d$, we have
\begin{align*}
b_{j,t}(\widetilde{\exp}_x(hy))&= b_{j,t}(x_0) + O(h^\alpha).
\end{align*}

To obtain a Taylor expansion for $e^{\frac{i}{h}\phi_{j,t}(\widetilde{\exp}_x(hy))}$, we write $\tilde{x}:= h^{-\alpha} \widetilde{\exp}_{x_0}^{-1}(x)$, so that $\tilde{x} \in B_{eucl}(0,1)$.
We first note that
$$\phi_{j,t}(x)= \phi_{j,t}(\widetilde{\exp}_{x_0}(h^{\alpha} \tilde{x})) = \phi_{j,t}(x_0) + h^\alpha \tilde{x} \cdot  {\color{black} \partial \left( \phi_{j,t}\circ \widetilde{\exp}_{x_0}  \right)(0) }  + O(h^{2\alpha}),$$
thanks to the definition of $\xi_j^{t, x_0}$. 

Using the fact that $\textcolor{black}{x}\mapsto \partial (\phi_{j,t}\circ(\widetilde{\exp}_{x}(z))|_{z=0}$ is $C^1$, we then have
\begin{align*}
\phi_{j,t}(\widetilde{\exp}_x(hy)) &= \phi_{j,t}(x) + hy \cdot \partial (\phi_{j,t}\circ(\widetilde{\exp}_x(z))|_{z=0}+ O(h^2)\\
&=\phi_{j,t}(x_0) + {\color{black}(h^\alpha \tilde{x}+ hy) \cdot  \partial \left( \phi_{j,t}\circ \widetilde{\exp}_{x_0}  \right)(0)} + O(h^{2\alpha}).
\end{align*}

All in all, we have
{\color{black}\begin{align*}
f_{x, h}^t(y) &\underset{\textup{by Proposition}\ref{prop:SumLag3}}{=}\sum_{j=1}^{M(t)} b_{j,t}(x_0) e^{i h^{-1}\phi_{j,t}(x_0) +i(h^\alpha \tilde{x}+ hy) \cdot  \partial \left( \phi_{j,t}\circ \widetilde{\exp}_{x_0}  \right)(0)} + O(h^{2\alpha - 1}) +O(h^\alpha)\\
&=\sum_{j=1}^{N(t)} B_{j,t}(x_0)e^{ \xi_j^{t,x_0} \cdot (h^{\alpha-1} \tilde{x}+ y)}+O(h^{2\alpha-1})+O(h^\alpha)\,\\
&=\sum_{j=1}^{N(t)}\beta_j^{t,x_0}e^{i\xi_j^{t,x_0}\cdot y+i\vartheta_j^{x_0,t}(\widetilde{x}; h)}+O(h^{2\alpha-1})+O(h^\alpha)\, ,
\end{align*}
where $\vartheta_j^{x_0,t}( \widetilde{x}; h)= \frac{1}{h} \mathrm{Arg}(B_{j,t}(x_0)) +h^{\alpha-1}\xi_j^{t,x_0}\cdot\widetilde{x}$, taking $\mathrm{Arg}(B_{j,t}(x_0))\in [0, 2\pi[$ to be the complex argument of $B_{j,t}(x_0)$.}

Since $\alpha>\frac{1}{2}$, the error terms vanish as $h\rightarrow 0$. Therefore, if we define the continuous function
\[\Gamma : \mathbb{T}^{N(t)} \ni (\theta_1,\dots,\theta_{N(t)}) \mapsto G \left(\sum_{j=1}^{N(t)}\beta_j^{t,x_0} e^{i\xi_j^{t,x_0}\cdot y_1+i\theta_j}, \dots, \sum_{j=1}^{N(t)}\beta_j^{t,x_0} e^{i\xi_j^{t,x_0}\cdot y_m+i\theta_j} \right)\, ,
\]
we have
\begin{equation}\label{eq:fixed_frame_local_limits_2}
F(f^t_{x,h})= \Gamma(\vartheta_1^{x_0,t}(\widetilde{x}; h),\dots,\vartheta_{N(t)}^{x_0,t}(\widetilde{x}; h)) +o_{h\to 0}(1)\, .
\end{equation}

\textbf{Third step: Computing the expectation}

To compute the expectation of this quantity, we note that $\widetilde{\textsc{x}}$ is a random variable on $B_{eucl}(0,1)$, whose density we denote by $\frac{1}{\textup{Vol}(B_{eucl}(0,1))}\rho_h(z)\mathrm{d}z$. Since $d_0\widetilde{\exp}_{x_0}$ is an isometry, we have that
\begin{equation}\label{eq:fixed_frame_local_limits_1}
|\rho_h(z)-1|\leq C h^\alpha
\end{equation}
for all $z\in B_{eucl}(0,1)$ for some $C<+\infty$ which depends only on $(X,g)$ and on the choice of frame $(V(x))_{x\in X}$.

Therefore, if $\textsc{z}$ denotes a uniform random variable on $B_{eucl}(0,1)$, we have
$$\E \left[F(f^t_{\textsc{x},h})\right] =\E \left[ \Gamma(\vartheta_1^{x_0,t}(\textsc{z}; h),\dots,\vartheta_{N(t)}^{x_0,t}(\textsc{z}; h))\right] + o_{h\to 0}(1).$$

To compute this expectation, we want to use a multidimensional Kronecker theorem, whose proof we recall.

Suppose first of all that $\Gamma$ is of the form $\Gamma (\theta_1, ..., \theta_{N(t)})= e^{2i\pi(n_1\theta_1+\cdots+n_{N(t)}\theta_{N(t)})}$, where $\boldsymbol{n}:=(n_1,\dots,n_{N(t)})\in\Z^{N(t)}\setminus\{0\}$. Let us write $\boldsymbol{\xi^{x_0,t}_n}:=\sum_j^{N(t)}n_j\xi^{x_0,t}$, which is non-zero since the $\xi^{x_0,t}_j$ are rationally independent. Therefore, we have

\begin{align*}
\E \left[ \Gamma(\vartheta_1^{x_0,t}(\textsc{z}; h),\dots,\vartheta_{N(t)}^{x_0,t}(\textsc{z}; h))\right] &={\color{black}e^{(2i\pi/h)\sum_{j=1}^{N(t)} \mathrm{Arg}(B_{j,t}(x_0))}}  \frac{1}{\textup{Vol}(B_{eucl}(0,1))}\int_{B_{eucl}(0,1)}e^{2i\pi h^{\alpha-1}\boldsymbol{\xi^{x_0,t}_n}\cdot z}\rho_h(z)\mathrm{d}z.
\end{align*}
But $\int_{B_{eucl}(0,1)}e^{2i\pi h^{\alpha-1}\boldsymbol{\xi^{x_0,t}_n}\cdot z}dz$ is the Fourier transform of the indicator of the unit ball evaluated at $h^{\alpha-1}\boldsymbol{\xi^{x_0,t}_n}$. Since $\alpha<1$ and $\boldsymbol{\xi^{x_0,t}_n}\neq 0$, this goes to zero as $h\to 0$.

For a general $\Gamma$, we may approach it uniformly by a trigonometric polynomial having the same mean (this is a consequence of Fejér's theorem), and we see from what precedes that only the constant term will give a non-vanishing contribution to the expectation as $h\to 0$. Therefore, we have 
$$\lim\limits_{h\to 0} \E \left[ \Gamma(\vartheta_1^{x_0,t}(\textsc{z}; h),\dots,\vartheta_{N(t)}^{x_0,t}(\textsc{z}; h))\right] = \int_{\T^{N(t)}}G(\theta_1,\dots, \theta_{N(t)}) \mathrm{d}\theta_1 \cdots\mathrm{d}\theta_{N(t)}\, .$$

This quantity is exactly $\mathbb{E}_{\mathbb{P}_{\boldsymbol{\beta},\boldsymbol{\xi}}}\left[ F\right]$, and the result follows.
}

\end{proof}

\begin{remark}
We used the fact that the $\xi_j^{x_0,t}$ are rationally independent only in the last {\color{black}step} of the proof. If they are not rationally independent, then the phases 
$$\left(\xi_1^{x_0,t}\cdot y+ \vartheta_1^{x_0,t}(h), \dots, \xi_{N(t)}^{x_0,t}\cdot y + \vartheta_{N(t)}^{x_0,t}(h)\right) $$ get equidistributed along an affine sub-torus of $\T^{N(t)}$. The linear part of this torus depends only on the $\xi^{x_0,t}_j$, and not on $h$. 
However, the affine torus depends on the $\left( \vartheta_1^{x_0,t}(h),\dots, \vartheta_{N(t)}^{x_0,t}(h)\right)$, so we do not have convergence to a measure independent of $h$ (and hence, existence of a pointwise local weak limit). However, we may extract subsequences $h_n$ such that $\left( \vartheta_1^{x_0,t}(h_N),\dots, \vartheta_{N(t)}^{x_0,t}(h_N)\right)$ converges. Doing so, we ensure the existence of pointwise local weak limits, even when the $\xi_j^{x_0,t}$ are not rationally independent. We will not use this construction in the sequel, since we don't want to extract subsequences.
\end{remark}

\subsection{Long time behaviour of local limits}\label{ss:long_time}
The aim of this section is to prove the following proposition, which is the last step in the proof of Theorem \ref{th:Pointwise}, except for the fact that $\mathcal{E}_{(\lambda_1,\lambda_2)}^{T, irr}$ is a residual set. Recall that we fixed a phase $\phi \in \mathcal{E}^{T,irr}_{(\lambda_1,\lambda_2)}(\Omega)$, and a function $a\in C_c^\infty(\Omega)$. We now also fix point $x_0$ such that the vectors $(\xi_j^{t,x_0})_{j=1,\dots,N_{x_0}(t)}$ are rationally independent for all $t\geq T_0(\phi)$. Recall the definition \eqref{eq:measure_def} of $\boldsymbol{\lambda_\mu}$ associated to some measure $\boldsymbol{\mu}$ and those of $\boldsymbol{\lambda}_{a,\phi}$ and $\boldsymbol{\lambda}_{\boldsymbol{\mu}_{a,\phi}}$ given just below \eqref{eq:measure_def}.

\begin{tcolorbox}[breakable, enhanced]
\begin{proposition}[Pointwise local limits at long time]\label{prop:LongTimeLim}
The measures $\mathbb{P}_{\boldsymbol{\beta^{t,x_0}}, \boldsymbol{\xi^{t,x_0}}}$ converge weakly, as $t\longrightarrow +\infty$, to $\P_{\mu_{\boldsymbol{a},\boldsymbol{\phi}}}$.
\end{proposition}
\end{tcolorbox}

This proposition follows from the following two lemmas, which we prove below.

\begin{tcolorbox}[breakable, enhanced]
\begin{lemma}[A criterion for convergence in long time]\label{lem:WeakConv2}
Let $0<\lambda_1<\lambda_2$.
Suppose that, for all $t\geq 0$, we have integers $N(t)$, directions $\boldsymbol{\xi^t}= (\xi_1^t,\dots, \xi_{N(t)}^t) \in \left(\R^d\right)^{N(t)}$ such that for $j=1,\dots,N(t)$, $\lambda_1\leq |\xi_j^t|\leq \lambda_2$, and amplitudes $\boldsymbol{\beta^t}= (\beta_1^{t},\dots, \beta_{N(t)}^t) \in [0, +\infty)^{N(t)}$, which satisfy the following conditions:
\begin{itemize}
\item $\nu_t:= \sum_{j=1}^{N(t)} (\beta^t_j)^2 \delta_{\xi^t_j}$ converges weakly to $\boldsymbol{\lambda_\mu}$ for some measure $\boldsymbol{\mu}$ supported on a $[\lambda_1,\lambda_2]$.
\item $\max_{j\in 1,\dots, N(t)} \beta_j^t \underset{t\to +\infty}{\longrightarrow} 0$.
\end{itemize} 
Then $\mathbb{P}_{\boldsymbol{\beta^t}, \boldsymbol{\xi^t}}$ converges weakly to $\mathbb{P}_{\boldsymbol{\mu}}$.
\end{lemma}
\end{tcolorbox}

\begin{tcolorbox}[breakable, enhanced]
\begin{lemma}[A local quantum ergodicity result]\label{lem:EquidSCPonct}
The measures $\sum_{j=1}^{N(t)} |\beta^{t,x_0}_j|^2 \delta_{\xi^{t,x_0}_j}$ on $\R^d$ converges weakly as $t\to +\infty$ to $\boldsymbol{\lambda}_{a,\phi}$.
\end{lemma}
\end{tcolorbox}

Let us start with the proof of Lemma \ref{lem:WeakConv2}.

\begin{proof}[Proof of Lemma \ref{lem:WeakConv2}]
For each $t\geq 0$, consider the random function $f^t(y):=\sum_{j=1}^{N(t)}\beta_j^t e^{i\xi_j^t\cdot y+i\vartheta_j^t}$, where for each $t$, the $\vartheta_j^t$ are independent random variables uniformly distributed on $[0,2\pi]$. Thus $f^t$ has law $\mathbb{P}_{\boldsymbol{\beta^t},\boldsymbol{\xi^t}}$. 

For any compact set $K\subset \R^d$ and any $k\in \N$, we have 
\[
\mathbb{E}_{\mathbb{P}_{\boldsymbol{\beta^t}, \boldsymbol{\xi^t}}}\left[ \|f\|_{H^k(K)}^2 \right]= \sum_{l=0}^k \sum_{j=1}^{N(t)} (\beta_j^t)^2 |\xi_j^t|^l \Vol(K)\leq \left(\textup{Vol}(K)\sum_{l=0}^k\lambda_2^l\right)\nu_t(\R^d)\, ,
\]
which is bounded independently of $t$, by assumption. We may therefore apply Lemma \ref{lem:CritMes} and Remark \ref{rem:CondEsperance} to prove the result.\\
 
To this end, we fix $y_1,\dots,y_k\in\R^d$ and we study the convergence of the vector $(f^t(y_1),\dots,f^t(y_k))$ as $t\rightarrow +\infty$. We wish to apply a multivariate Lindeberg Central Limit Theorem to the sum over $j$ of the random vectors $\eta_j(t)=(\beta_j^te^{i\xi_j^t\cdot y_1+i\vartheta_j},\dots,\beta_j^te^{i\xi_j^t\cdot y_k+i\vartheta_j})$. By construction, the $\eta_j$'s are mutually independent. Moreover, for each $t\geq 0$ and $j\in\{1,\dots,N(t)\}$, $\E[\eta_j(t)]=0$ and the covariance of $\eta_j(t)$ has coefficients $\E[\eta_j^h(t)\overline{\eta_j^l(t)}]=(\beta_j^t)^2e^{i\xi_j^t(y_h-y_l)}$. Thus, the sum of their covariance matrices $M^t=(m^t_{hl})_{hl}$ has coefficients
\[
m^t_{hl}=\sum_{j=1}^{N(t)}(\beta_t^j)^2e^{i\xi_j^t\cdot(y_h-y_l)}
\]
which converges to $m_{hl}=\int_{\R^d} e^{i\xi(y_h-y_l)}d\boldsymbol{\lambda_\mu}(\xi)$ by the first assumption of the lemma. But the matrix $(m_{hl})_{hl}$ thus constructed is the covariance matrix of the random vector $(f(y_1),\dots,f(y_k))$ where $f$ is a random function following the law $\P_{\boldsymbol{\mu}}$. In particular, the matrix $M^t$ is invertible for all large enough $t$. Lastly, since $\sup_j \beta_j^t \xrightarrow[t\rightarrow 0]{}0$, we have (deterministically) $\sup_j|\eta(t)|=o(N(t))$, which implies the remaining condition for the multivariate Lindeberg Central Limit Theorem\footnote{Thanks to the Cramér-Wold Theorem \cite[Theorem 29.4]{Bili}, the multivariate Lindeberg Central Limit Theorem follows from the usual Lindeberg Central Limit Theorem \cite[Chapter 27]{Bili}}. Thus, as $t\rightarrow +\infty$, the vector $(f^t(y_1),\dots,f^t(y_k))$ converges in law to a Gaussian vector $(\zeta_1,\dots,\zeta_k)$ with covariance $(m_{hl})_{hl}$. We may then conclude thanks to Lemma \ref{lem:CritMes} and Remark \ref{rem:CondEsperance}.
\end{proof}

Before proceeding with the proof of Lemma \ref{lem:EquidSCPonct}, let us introduce some notations.\\

Recall that $V=(V_1,\dots,V_d)$ is an orthonormal frame defined in a neighbourhood $\calU$ of $x_0$. Using the Riemannian metric, it naturally induces an orthonormal co-frame $(V_1^*,\dots,V_d^*)$, that is to say a family of smooth sections $(V_i^*)_{i=1,\dots,d} : X\longrightarrow T^*X$ such that, for each $x\in X$, $(V^*_1(x),\dots,V^*_d(x))$ is an orthonormal basis of $T_x^*X$. If $x\in \mathcal{U}$ and $y\in \R^d$, we will write $yV^*(x) := y_1 V^*_1(x)+\dots +y_d V^*_d(x) \in T_x^*X$. Conversely, if $\xi \in T_x^*X$, we write $(V_x^*)^{-1}(\xi)$ for the unique $y\in \R^d$ such that $yV^*(x)= \xi$. We refer the reader to section \ref{sec:semi-classique} for the definition and standard results regarding semi-classical measures, which we use in the proof. Recall also that $\Phi^t:T^*X\rightarrow T^*X$ is the geodesic flow.

\begin{proof}[Proof of Lemma \ref{lem:EquidSCPonct}]
The sequence $(a e^{i\phi/h})_{h>0}$ has a semi-classical measure, which we denote by $\nu_0$. By Egorov's theorem (Theorem \ref{t:Egorov} below), the semi-classical measure of $f_h^t=U_h^t (a e^{i\phi/h})$ is $\nu_t = \Phi^t_* \nu_0$. By \cite[Theorem 1]{Schu}, if $\mathrm{Liou}_\lambda$ denotes the Liouville measure on $S_\lambda^*$, then $\nu_t$ converges weakly to the measure $ \boldsymbol{\nu}_{a,\phi}:=\int_{\lambda_1}^{\lambda_2}{\color{black} \mathrm{Liou}_\lambda} \mathrm{d}\boldsymbol{\mu}_{a,\phi}$. Let $\varepsilon>0$, and $\chi_1 \in C_c^\infty(X)$ be supported in a neighbourhood of size $\varepsilon$ of $x_0${\color{black}, such that $\int_X\chi_1(x)dx=1$}. Let $\chi_2\in C_c^\infty(\R^d)$. We define $\chi\in C_c^\infty(T^*X)$ by

\[
\chi(x,\xi) = \chi_1(x) \chi_2\left( \left(V^*_x\right)^{-1}(\xi)\right)\, .
\]

By the previous remarks, we have 
\begin{align*}
\int_{T^*X} \chi(x,\xi) \mathrm{d}\nu_t(x,\xi) \underset{t\to +\infty}{\longrightarrow} &\int_{T^*X} \chi(x,\xi)\mathrm{d}\boldsymbol{\nu}_{a,\phi}(x,\xi)\\
&= \int_{S^{d-1}} \int_{\lambda_1}^{\lambda_2} \chi_2(\lambda v)   \mathrm{d}\boldsymbol{\mu}_{a,\phi}(\lambda) \mathrm{d}v + O(\varepsilon)\\
&= \int_{\mathbb{R}^{d-1}} \chi_2(w)   \mathrm{d}\boldsymbol{\lambda}_{a,\phi}(w) + O(\varepsilon)\, ,
\end{align*}
with $\boldsymbol{\lambda}_{a,\phi}$ and $\boldsymbol{\mu}_{a,\phi}$ as in section \ref{subsec:GaussField}. On the other hand, by Proposition \ref{prop:SumLag3} we know that, as $h\to 0$,
\[f_h^t (x) = \sum_{j=1}^{{\color{black}M(t)}}  e^{i\phi_{j, t}(x)/h} b_{j, t}(x) + o_{L^2}(1),
\]
so that, by \eqref{eq:semi-classical_measure} and the $L^2$-continuity of semi-classical measures (which follows for instance from Theorem 5.1 of \cite{Zworski_2012}),
\begin{align*}
\int_{T^*X} &\chi(x,\xi)\mathrm{d}\nu_t(x,\xi) = \int_{T^*X} \sum_{j=1}^{N(t)} \chi_1(x)\chi_2\left((V_x^*)^{-1}\left(\partial \phi_{j,t}(x)\right)\right) |{\color{black}B}_{j,t}(x)|^2 \mathrm{d}x\\
&\overset{\eqref{eq:phase_bound}}{=}\int_{T^*X} \sum_{j=1}^{N(t)} \chi_1(x)\chi_2\left((V_{x_0}^*)^{-1}\left(\partial \phi_{j,t}(x_0)\right)\right) |{\color{black}B}_{j,t}(x)|^2 \mathrm{d}x+O(\eps)\int_{T^*X} \sum_{j=1}^{N(t)} \chi_1(x)|{\color{black}B}_{j,t}(x)|^2 \mathrm{d}x\\
&\overset{\eqref{eq:LaBorneLog}}{=}\sum_{j=1}^{N(t)}  |{\color{black}B}_{j,t}(x_0)|^2 \chi_2\left( (V_{x_0}^*)^{-1}\left(\partial \phi_{j,t}(x_0)\right)\right) (1+O(\eps))+O(\eps)\nu_t(\chi_1)\\
&\overset{\eqref{eq:LaBorneLog}}{=}\sum_{j=1}^{N(t)}  |{\color{black}B}_{j,t}(x_0)|^2 \chi_2\left( (V_{x_0}^*)^{-1}\left(\partial \phi_{j,t}(x_0)\right)\right) (1+O(\eps))+O(\|\chi_1\|_\infty\eps)\, .
\end{align*}

To obtain the second line, we used the smoothness of the vector fields $V_i^*$ and of $\chi_2$. Note that $ (V_{x_0}^*)^{-1}\left(\partial \phi_{j,t}(x_0)\right) = \xi_j^t$ and recall that ${\color{black}|B_{j,t}(x_0)|}=\beta_j^t$. For the last line, we use the fact that, since $\nu_t=\Phi^t_*\nu_0$ as observed at the start of the proof, the total mass of $\nu_t$ is constant. We deduce that
\[
\limsup\limits_{t\to +\infty} \left|\sum_{j=1}^{N(t)}|\beta_j^t|^2 \chi_2(\xi_j^t) - \int_{\mathbb{R}^{d-1}} \chi_2(w)   \mathrm{d}\boldsymbol{\lambda}_{a,\phi}(w) \right| = O(\varepsilon),
\]
so that  $\sum_{j=1}^{N(t)}(\beta_j^t)^2 \chi_2(\xi_j^t) \longrightarrow \int_{\mathbb{R}^{d-1}} \chi_2(w)   \mathrm{d}\boldsymbol{\lambda}_{a,\phi}(w)$. In other words, $\sum_{j=1}^{N(t)}  (\beta_j^t)^2 \delta_{\xi_j^t}$ converges weakly to $\boldsymbol{\lambda}_{a,\phi}$.
\end{proof}

\subsection{Conclusion of the proofs}\label{ss:conclusion}

In this section we use the results from sections \ref{ss:criterion}, \ref{ss:propagation}, \ref{ss:convergence} and \ref{ss:long_time}, as well as Propositions \ref{prop:polychrom_conclusion} and \ref{prop:monochrom_conclusion} from the following section, to prove Theorems \ref{th:Pointwise} and \ref{theoMono}.

\begin{proof}[Proof of Theorem \ref{th:Pointwise}]
Let $\calE^{T,irr}_{(\lambda_1,\lambda_2)}(\Omega)$ be as in \eqref{eq:DefIrr}, which is a residual subset of $\calE^T_{(\lambda_1,\lambda_2)}(\Omega)$ by Proposition \ref{prop:polychrom_conclusion}. Let $a\in C^\infty_c(\Omega)$, let $\phi\in\calE^{T,irr}_{(\lambda_1,\lambda_2)}(\Omega)$. Then, there exists $T_0(\phi)<+\infty$ such that for almost every $x_0\in\calU$, for every $t\geq T_0$, the vectors $(\xi_j^{t,x_0})_{j=1,\dots,N_{x_0}(t)}$, defined in \eqref{eq:xi_beta}, are rationally independent. Let $\boldsymbol{\xi^{t,x_0}}$ and $\boldsymbol{\beta^{t,x_0}}$ be as in \eqref{eq:xi_beta}. Then, by Proposition \ref{prop:LWSumLag}, the field $(f_h^t)_h$ has an $h^\alpha$-pointwise local weak limit at $x_0$ given by $\mathbb{P}_{\boldsymbol{\beta^{t,x_0}},\boldsymbol{\xi^{t,x_0}}}$ (from Definition \ref{def:rpws}). Next, by Proposition \ref{prop:LongTimeLim}, the measures $\mathbb{P}_{\boldsymbol{\beta^{t,x_0}},\boldsymbol{\xi^{t,x_0}}}$ converge to $\mathbb{P}_{\boldsymbol{\mu_{a,\phi}}}$ (defined in section \ref{subsec:GaussField}).
\end{proof}
\begin{proof}[Proof of Theorem \ref{theoMono}]
The proof is very close to that of Theorem \ref{th:Pointwise}. The only differences are the following. The set $\calE^{T,irr}_{(\lambda_1,\lambda_2)}(\Omega)$ should be replaced by $\calC^{T,irr}(\Sigma)$ and Proposition \ref{prop:polychrom_conclusion} should be replaced by Proposition \ref{prop:monochrom_conclusion}. For the rest of the proof, one takes $u\in\calC^{T,irr}(\Sigma)$, which induces a phase $\phi_u$ defined on an open subset $\Omega_u$. The rest of the proof carries over with $\phi_u$ (resp. $\Omega_u$) in place of $\phi$ (resp. $\Omega$).
\end{proof}

\section{Classical and quantum dynamics of Lagrangian submanifolds}\label{sec:ClasDyn}

The aim of this section is to prove Proposition \ref{prop:SumLag3}. In sections \ref{ss:back_to_hyperbolicity} and \ref{ss:on_lagrangian_submfs} we introduce basic definitions and properties related to the hyperbolic dynamics on $S^*X$. In section \ref{ss:ClassicHadamard}, we then apply these to state Lemma \ref{lem:LaBorneChiante}, which is a key estimate needed in the proof (more precisely, we need it to prove \eqref{eq:LaBorneLog}). In section \ref{ss:labornechiante}, we prove Lemma \ref{lem:LaBorneChiante}. Finally, in section \ref{ss:demo_propagation}, we prove Proposition \ref{prop:SumLag3}. In all this section, we fix an arbitrary metric $g_0$ on $T^*X$. 

\subsection{Hyperbolicity}\label{ss:back_to_hyperbolicity}

For each $\lambda > 0$, we denote by $\Phi_\lambda^t:S_\lambda^*X\rightarrow S_\lambda^*X$, $t\in\R$ the geodesic flow on $S_\lambda^*X$.
Since $X$ has negative curvature, $(\Phi_\lambda^t)_t$ is an Anosov flow (see \cite{Ebe} for a proof of this fact). It means that for each $\rho\in S_\lambda^*X$, there exist $E_\rho^+$, $E_\rho^-$ and $E_\rho^0$ subspaces of $T_\rho S_\lambda^*X$, respectively called the \emph{unstable}, \emph{stable} and \emph{neutral} direction at $\rho$ such that:
\begin{itemize}
\item $T_\rho S_\lambda^*X=E_\rho^+\oplus E_\rho^-\oplus E_\rho^0$.
\item The distributions $E_\rho^+$, $E_\rho^-$ and $E_\rho^0$ depend H\"older continuously on $\rho$. 
\item The distribution $E_\rho^0$ is one dimensional and generated by $\frac{\mathrm{d}}{\mathrm{d}t}|_{t=0}\Phi^t_\lambda(\rho)$. In particular, $\mathrm{d}\Phi^t_\lambda|_{E^0}$ is bounded from above and below uniformly in $t$.
\item $E_\rho^+$ and $E_\rho^-$ are both $d-1$ dimensional, and for each $t\in\R$, we have \begin{equation}\label{eq:DirInv}
\mathrm{d}_\rho\Phi^t_\lambda(E^\pm_\rho)=E^\pm_{\textcolor{black}{\Phi^t_\lambda}(\rho)}.
\end{equation}
\item There exists $C>0$ and $A>1$ such that for each $\rho\in S_\lambda^*X$, $t>0$, $\xi^+\in E^+_\rho$ and $\xi^-\in E^-_\rho$,
\begin{equation}\label{eq:DefHyp}
\begin{aligned}
\|\mathrm{d}_\rho\Phi^{-t}_\lambda(\xi^+)\|_{g_0}&\leq C A^{-t} \|\xi^+\|_{g_0}\\
\|\mathrm{d}_\rho\Phi^t_\lambda(\xi^-)\|_{g_0}&\leq C A^{-t}\|\xi^-\|_{g_0}\, .
\end{aligned}
\end{equation}
\end{itemize}

%Furthermore, since $\Phi_\lambda^t$ is the geodesic flow on a manifold of negative curvature, if $\pi:S_\lambda^*X\rightarrow X$ is the canonical projection,  we have that $\textup{Ker}(d\pi)\cap (E_\rho^+\oplus E_\rho^0)=\{0\}$.\\

If $\rho = (x,\xi) \in S_\lambda^*X$ for some $\lambda>0$, let us write $\hat{E}^0_\rho :=\{(0,s\xi)\, :\, s\in\R\}$ where $\rho=(x,\xi)$. Note that $T_\rho T^*X = E_\rho^+\oplus E_\rho^0\oplus E_\rho^- \oplus \hat{E}_\rho^0$. 
In a basis adapted to this decomposition, we have
\begin{equation}\label{eq:Matrix}
\mathrm{d}_\rho \Phi^t = \begin{pmatrix}
M_{\rho,t} & 0 & 0 & 0\\
0 & 1 & 0 & \lambda t\\
0 & 0 & (M_{\rho,t}^{-1})^\dagger & 0\\
0 & 0 & 0 & 1
\end{pmatrix},
\end{equation}
where $M_{\rho,t}$ is a $(d-1) \times (d-1)$ matrix such that $\|M_{\rho,t} \xi\| \geq C A^t \|\xi\|$ for any $\xi \in \R^{d-1}$. It follows from (\ref{eq:DirInv}) and (\ref{eq:Matrix}) that $E_\rho^-\oplus \hat{E}^0_\rho$ and $E_\rho^+\oplus E_\rho^0$ are Lagrangian spaces. 

\textcolor{black}{If $\sigma$ denotes the canonical symplectic structure on $T^*X$, we may find a constant $C_0>0$ such that, for all $\rho\in T^*X$ and all $\xi, \zeta \in T_\rho T^*X$, we have}
\begin{equation}\label{eq:CestRelouDeReprendreLePapier}
\textcolor{black}{|\sigma(\xi, \zeta)| \leq C_0 \|\xi\|_{g_0} \|\zeta\|_{g_0}.}
\end{equation}

\textcolor{black}{Furthermore, the map $\Phi^t$ being symplectic, we have $\sigma(\xi, \zeta) = \sigma( \mathrm{d}_\rho \Phi^t(\xi),  \mathrm{d}_\rho \Phi^t(\zeta))$. Combining this with (\ref{eq:Matrix}) and letting $t\to \pm \infty$, we see that} 
\begin{equation}\label{eq:CestRelouDeReprendreLePapier2}
\textcolor{black}{\left(\xi \in E_\rho^\pm \text{ and } \zeta \in  E_\rho^\pm \oplus E_\rho^0 \oplus \hat{E}_\rho^0 \right)  \Longrightarrow \sigma(\xi, \zeta)=0.}
\end{equation}

\textcolor{black}{In {\color{black}particular}, $E_\rho^0 \oplus \hat{E}_\rho^0$ is symplectically orthogonal to $E_\rho^+ \oplus E_\rho^-$. Since $E_\rho^0 \oplus \hat{E}_\rho^0$ forms a vector space of dimension 2, and there is a unique symplectic form on $\R^2$ up to a multiplicative constant, if $\xi = (\xi^0, \hat{\xi}^0) \in E_\rho^0 \oplus \hat{E}_\rho^0$ and $\zeta = (\zeta^0, \hat{\zeta}^0) \in E_\rho^0 \oplus \hat{E}_\rho^0$, we have $\sigma(\xi, \zeta) = c(\rho) \left( \xi^0 \hat{\zeta}^0 -  \zeta^0 \hat{\xi}^0\right)$. By continuity and compactness, if $0<\lambda_1 \leq \lambda_2$, there exists $0< c_1<c_2$ such that}
\begin{equation}\label{eq:CestRelouDeReprendreLePapier3}
\textcolor{black}{\rho \in S^*_{[\lambda_1, \lambda_2]} \Longrightarrow c_1 \leq |c(\rho)| \leq c_2}.
\end{equation}

Finally, we define the stable and weak stable manifolds of $\rho$ as
\begin{align*}
W^{-}(\rho) &= \{ \rho'\in S_\lambda^*X ~|~ d (\Phi_\lambda^t \rho, \Phi_\lambda^t \rho') \underset{t\longrightarrow +\infty}{\longrightarrow} 0\}\\
W^{-0}(\rho) &= \{ \rho'\in S_\lambda^*X ~|~ d (\Phi_\lambda^t \rho, \Phi_\lambda^t \rho') \text{ remains bounded as } t\longrightarrow +\infty\}.
\end{align*}

$W^{-}(\rho)$ and $W^{-0}(\rho)$ are then manifolds, whose tangent space at $\rho$ are respectively $E_\rho^-$ and $E_\rho^- \oplus E^0_\rho$. Furthermore, if $\rho'\in W^{-0}(\rho)$, there exists $t\in \R$ such that $\textcolor{black}{\Phi^t_\lambda}(\rho') \in W^-(\rho)$.

\subsection{Properties of Lagrangian submanifolds of $T^*X$}\label{ss:on_lagrangian_submfs}
In this section we introduce some basic properties of Lagrangian submanifolds of \textcolor{black}{$T^*Y$, where $Y$ is a Riemannian manifold. Recall that a Lagrangian submanifold is a submanifold $\Lambda\subset T^*Y$ of dimension $d$, such that the canonical symplectic form of $T^*Y$ vanishes on $T_\rho \Lambda$ for any $\rho \in \Lambda$ (see \cite[Chapter 1]{DiSj}). Here, we will focus on a special family of Lagrangian submanifolds, which can be written as graphs.}
\begin{tcolorbox}[breakable, enhanced]
\begin{definition}
\begin{itemize}
\item Let $Y$ be a smooth \textcolor{black}{Riemannian} manifold. We say that a Lagrangian submanifold $\Lambda\subset T^*Y$ is \emph{projectable} if there exist an open subset $\Omega_\Lambda\subset Y$ and a smooth real-valued function $\phi$ defined on a neighbourhood of $\overline{\Omega}_\Lambda$ such that
\[
\Lambda=\{(x,\partial\phi(x))\, :\, x\in\Omega_\Lambda\}\,  .
\]
\item We {\color{black}call} $\phi$ a \emph{phase function} and {\color{black}say} that it \emph{generates} $\Lambda$. Note that $\phi$ is monochromatic if and only if $\Lambda\subset S^*Y$. We call $\Omega_\Lambda$ the \emph{support} of $\Lambda$.
\item Given also $\Lambda'\subset T^*Y$ a Lagrangian submanifold, we say that $\Lambda'$ is a \emph{Lagrangian extension} of $\Lambda$ {\color{black}if} $\overline{\Lambda}\subset\Lambda'$.
\end{itemize}
\end{definition}
\end{tcolorbox}
\begin{remark}\label{Rem:CriterionProj}
Let $\Lambda$ be a submanifold of $T^*Y$. Then, $\Lambda$ is a projectable Lagrangian manifold if and only if $\Lambda$  is the graph of a smooth section of $T^*Y$ defined over an open subset $\Omega_\Lambda$, which {\color{black}can be extended} smoothly to some neighbourhood of $\overline{\Omega}_\Lambda$.
\end{remark}

\begin{tcolorbox}[breakable, enhanced]
\begin{definition}\label{def:transverse_to_stable_directions}
Let $0<\lambda_1<\lambda_2$, and let $\Lambda \subset S^*_{[\lambda_1,\lambda_2]}X$ be some Lagrangian submanifold. We say that $\Lambda$ is \emph{transverse to the stable directions} if it admits a Lagrangian extension $\Lambda'$ such that for any $\rho\in \Lambda'$,  we have
$$T_\rho \Lambda' \cap E_\rho^-  = \{0\}.$$
\end{definition}
\end{tcolorbox}
Note that, if $\Lambda\subset S^*_\lambda X$ for some $\lambda>0$, $T_\rho\Lambda'\cap E^-_\rho=\{0\}$ is equivalent to $T_\rho\Lambda'\oplus E^-_\rho=T_\rho(S^*_\lambda X)$. In the case where $\Lambda$ is {\color{black} a section of $T^*X$}, this is equivalent to the fact that this section is transverse at $\rho$ to the unique stable manifold $W^-(\rho')$ containing $\rho$. This motivates our use of the term \textit{transverse} in this context.

\begin{tcolorbox}[breakable, enhanced]
\begin{definition}
Let $0<\lambda_1<\lambda_2$, and let $\Lambda \subset S^*_{[\lambda_1,\lambda_2]}X$ be some Lagrangian submanifold. 
We say that $\Lambda$ is \emph{nowhere stable} if it admits a simply connected Lagrangian extension $\Lambda'$ such that for any $\rho_1, \rho_2\in \Lambda'$, we have
\begin{equation}\label{eq:DefNonStable}
\left(\rho_2 \in W^{-}(\rho_1)\right) \Longrightarrow \left(\rho_2=\rho_1\right).
\end{equation}
\end{definition}
\end{tcolorbox}

\begin{tcolorbox}[breakable, enhanced]
\begin{lem}\label{lem:PetitesLag}
Let $Y$ be a smooth \textcolor{black}{Riemannian} manifold. Let $\Lambda \subset S^*_{[\lambda_1,\lambda_2]}{\color{black}Y}$ be a precompact Lagrangian submanifold transverse to the stable directions. Then there exists finitely many Lagrangian submanifolds $\Lambda_1,\dots, \Lambda_n$ such that $\Lambda= \cup_{i=1}^n \Lambda_i$ and each $\Lambda_i$ is nowhere stable.
\end{lem}
\end{tcolorbox}

\begin{proof}
Let $\Lambda'$ be a Lagrangian extension of $\Lambda$. By our transversality assumption, we know that \textcolor{black}{the points of $\Lambda'\cap W^-(\rho)$ are isolated. In other words, for any $\rho\in \Lambda'$, there exists $\varepsilon_\rho>0$ such that $\Lambda'\cap W^-(\rho)\cap B(\rho, \varepsilon_\rho) = \{\rho\}$, where $B(\rho, \varepsilon_\rho)$ denotes the open ball of center $\rho$ and of radius $\varepsilon_\rho$.}

Since $\Lambda'$ is a smooth manifold and the dependence of the unstable directions in $\rho$ is Hölder, we see that $\rho\mapsto \varepsilon_\rho$ is continuous. Hence $\varepsilon_0:= \inf_{\rho\in \overline{\Lambda}} \varepsilon_\rho$ is $>0$.

\textcolor{black}{Let us consider a covering of $\Lambda'$ by finitely many balls of radius $\frac{\varepsilon_0}{2}$, and check that each element of this covering is nowhere stable. If $\rho_1, \rho_2$ belong to the intersection of $\Lambda'$ with a ball of radius $\frac{\varepsilon_0}{2}$, then we have $\rho_2\in B(\rho_1, \varepsilon_0) \cap \Lambda\subset B(\rho_1, \varepsilon_{\rho_1}) \cap \Lambda$. Therefore, we have $(\rho_2\in W^-(\rho_1)) \Longrightarrow (\rho_2=\rho_1)$, as announced.}
\end{proof}

\subsection{Evolution of Lagrangian manifolds on Hadamard manifolds}\label{ss:ClassicHadamard}

Next we will focus on the evolution of nowhere stable Lagrangian submanifolds on the universal cover of $X$, which we denote by $\tilde{X}$. The manifold $\tilde{X}$ is then a Hadamard manifold, i.e., a complete simply connected manifold of negative curvature. In particular, we state the key estimate Lemma \ref{lem:LaBorneChiante} needed in the proof of Proposition \ref{prop:SumLag3}.

\begin{tcolorbox}[breakable, enhanced]
\begin{definition}
Let $Y$ be a Riemannian manifold and let $\pi:T^*Y\rightarrow Y$ be the canonical projection. Let $\Lambda\subset T^*Y$ be a Lagrangian submanifold. Let $(\Phi^t)_t$ be the geodesic flow, acting on $T^*Y$. We say that $\Lambda$ is \emph{expanding} if there is a Lagrangian extension $\Lambda'$ of $\Lambda$ such that for any $\rho,\rho'\in\Lambda'$ which do not belong to the same geodesic, the function $t\mapsto\textup{dist}_Y(\pi_Y\Phi^t(\rho),\pi_Y \Phi^t(\rho'))$ is increasing.
\end{definition}
\end{tcolorbox}

\begin{tcolorbox}[breakable, enhanced]
\begin{lem}\label{lem:LongTimeProj}
Let $Y$ be a complete simply connected manifold of negative curvature, and let $\Lambda\subset S^*_{[\lambda_1,\lambda_2]}{\color{black}Y}$ be a Lagrangian submanifold which is nowhere stable. Then there exists $T_0(\Lambda)>0$ such that for all $t\geq T_0$, $\Phi^{t} (\Lambda)$ is an expanding projectable Lagrangian submanifold.
\end{lem}
\end{tcolorbox}

\begin{proof}
Let $\Lambda', \Lambda''$ be Lagrangian extensions of $\Lambda$ with $\overline{\Lambda'}\subset \Lambda''$, both satisfying (\ref{eq:DefNonStable}).\\

Let $\rho_1,\rho_2\in \Lambda''$ be points which do not belong to the same geodesic. By the proof of \cite[Theorem 4.8.2]{Jost}, there exists $c>0$ such that $ \frac{\mathrm{d}^2}{\mathrm{d} t^2} \textup{dist}_Y^2(\Phi^t \rho_1,\Phi^t \rho_2) \geq c~ \textup{dist}_Y(\Phi^t \rho_1,\Phi^t \rho_2)$. In particular, when $t\to +\infty$, $\textup{dist}_Y(\Phi^t \rho_1,\Phi^t \rho_2) $ either converges to zero or diverges to $+\infty$.\\

Suppose that this map converges to zero as $t\to +\infty$, so that it is decreasing. Then we must also have $\mathrm{dist}_{T^*Y}(\Phi^t \rho_1,\Phi^t \rho_2)$ converging to zero. Indeed, if this were not the case, we could find large times $t$ at which the points $\Phi^t \rho_1$ and $\Phi^t \rho_2$ are very close when projected on $Y$, but have directions which are not close to each other. The distance on the base of such points cannot be a decreasing function. Therefore, we must have $\rho_2\in W^-({\color{black}\rho_1)}$, {\color{black}which contradicts the fact that $\Lambda$ is nowhere stable}.\\

Hence, we must have $\textup{dist}_Y(\pi_Y(\Phi^t \rho_1),\pi_Y(\Phi^t \rho_2))\to +\infty$ as $t\to+\infty$, so that the distance between $\Phi^t \rho_1$ and $\Phi^t \rho_2$ will be increasing after a time $T(\rho_1,\rho_2)$ where it is minimal. This time $T(\rho_1, \rho_2)$ depends continuously on $\rho_1,\rho_2$, so, by compactness, we can find $T_0$ such that for all $t\geq T_0$ and all $\rho_1,\rho_2 \in \Lambda'$, we have $\textup{dist}_Y(\pi_Y(\Phi^t \rho_1),\pi_Y(\Phi^t \rho_2))>0$, and this quantity is increasing with $t$. In particular, $\Phi^t(\Lambda')$ is a smooth section of $T^*Y$, so that it can be put in the form $\{(x, \theta_{t}(x)) : x\in \Omega'_{t} \}$. Since $\Lambda'$ is simply connected, so is $\Phi^t(\Lambda')$ and therefore $\Omega'_{t}$ is simply connected. Therefore, $\theta_t$ can be chosen as the differential of some function $\phi_t$, so that $\Phi^t(\Lambda)$ is a projectable expanding Lagrangian submanifold of $T^*Y$.
\end{proof}

For all $t \in \{0\}\cup [T_0(\Lambda), +\infty)$, let us denote by $\Omega_t$ the support of $\Phi^t(\Lambda)$, and by $\phi_{t}\in C^\infty(\Omega_{t})$ a generating function for $\Phi^{t}(\Lambda)$.
 Let $t_1\in \{0\}\cup [T_0(\Lambda), +\infty)$, and $t_2 \geq T_0(\Lambda)$. Since $\Phi^{t_1}(\Lambda)$ and $\Phi^{t_2}(\Lambda)$ are projectable, the map $\kappa_{t_1,t_2}:\Omega_{t_1}\ni x \mapsto\pi_Y(\Phi^{t_2-t_1}(x,\partial\phi_{t_1}(x)))\in\Omega_{t_2}$ is then an embedding, and we will write $\kappa_{t_1,t_2}^{-1} =: g_{t_1,t_2}:\Omega_{t_2}\rightarrow\Omega_{t_1}$. Therefore, for all $y\in \Omega_{t_2}$, we have
\begin{equation}\label{eq:DefG}
(y,\partial\phi_{t_2}(y))=\Phi^{t_2-t_1}\left(g_{t_1,t_2}(y),\partial\phi_{t_1}(g_{t_1,t_2}(y))\right)\, .
\end{equation}

Note that for any $t_1,t_2,t_3 \geq T_0$, we have
\begin{equation}\label{eq:Fonctorialite}
g_{t_1,t_2} \circ g_{t_2,t_3} = g_{t_1,t_3}.
\end{equation}

The following lemma, which we prove in the next section, gives us an estimate on the regularity of the maps $g_{t_1,t_2}$ which will be essential {\color{black}to obtain the first point in Proposition \ref{prop:SumLag3}}.

\begin{tcolorbox}[breakable, enhanced]
\begin{lemma}\label{lem:LaBorneChiante}
Let $\Lambda\subset S_{[\lambda_1,\lambda_2]}^*X$ be a Lagrangian manifold which is transverse to the stable directions, and let $T_0(\Lambda)$ be as in Lemma \ref{lem:LongTimeProj}.
Then for all $t\geq T_0(\Lambda)$, $\log(|\det(\mathrm{d}_x g_{0,t})|)$ is continuous in $x$, uniformly in $(x,t)$.
\end{lemma}
\end{tcolorbox}

\subsection{Proof of Lemma \ref{lem:LaBorneChiante}}\label{ss:labornechiante}
In this section, we prove Lemma \ref{lem:LaBorneChiante} but before doing so, we state and prove a final auxiliary lemma. Recall that we fixed a metric on $T^*X$, which allows us to define angles between vectors of $T_\rho T^*X$ for any $\rho\in S_{[\lambda_1,\lambda_2]}^*X$.

\begin{tcolorbox}[breakable, enhanced]
\begin{definition}\label{def:EtaTransversalite}
Let $\eta_0>0$. We say that a Lagrangian submanifold $\Lambda\subset S_{[\lambda_1,\lambda_2]}^*X$ is \emph{$\eta_0$- transverse to the stable directions} if, for any $\rho\in \Lambda$,  the angle between any vector of $T_\rho\Lambda$ and any vector of $E_\rho^-$ is at least $\eta_0$.
\end{definition}
\end{tcolorbox}

\begin{tcolorbox}[breakable, enhanced]
\begin{lemma}\label{l:charts}
For each $\eta_0>0$ and $\delta_0>0$, there exists $T_1=T_1(\eta_0,\delta_0)<+\infty$ such that the following holds. 

Let $\Lambda\subset S_{[\lambda_1,\lambda_2]}^*X$ be a Lagrangian manifold which is $\eta_0$-transverse to the stable directions.  Then for each $t\geq T_1$ and each $\rho\in \Phi^t(\Lambda)$, the angle between the space of $T_\rho\Phi^t(\Lambda)$ and $E_\rho^+\oplus E_\rho^0\oplus \hat{E}_\rho^0$ is smaller than $\delta_0$.

More precisely for each $t\geq T_1$ and each $\rho\in \Phi^t(\Lambda)$, we may find a vector space $\tilde{E}^0_\rho \subset E_\rho^0\oplus \hat{E}_\rho^0$, whose dependence on $\rho$ is Hölder, such that
the angle between $T_\rho\Phi^t(\Lambda)$ and $E_\rho^+\oplus \tilde{E}_\rho^0$ is smaller than $\delta_0$.
\end{lemma}
\end{tcolorbox}

\begin{proof}
First of all, note that there exists $c>0$ such that for all $\rho\in S^*_{[\lambda_1,\lambda_2]} X$ and all $\xi=(\xi^+,\xi^-,\xi^0, \hat{\xi}^0)\in E^+_{\rho}\oplus E^-_{\rho}\oplus E^0_{\rho}\oplus \hat{E}^0_{\rho}$, we have
\begin{equation}\label{eq:EquivNorm}
 \|\xi\|^2 \geq c \|\xi^-\|^2, 
 ~~~~\|\xi\|^2 \geq c \left( \|\xi^+\|^2 + |\xi^0|^2 + |\hat{\xi}^0|^2 \right).
\end{equation}
For a given $\rho$, this follows from the fact that all norms are equivalent on a finite-dimensional space, and the constant $c$ involved depends on the angle between the directions $E^0_\rho$, $\hat{E}^0_\rho$ and $E^+_\rho$. By compactness, the constant may hence be taken independent of $\rho$.

Let us fix $\eta_0>0$ and $\Lambda$ as in the statement. Let $\rho\in \Lambda$ and let $t\in \R$. Write $\rho_t := \Phi^t(\rho)$ and $\xi_t := (d_\rho \Phi^t)(\xi)\in T_{\rho_t} \Phi^t(\Lambda)$. Decomposing $\xi_t$ as $\xi_t=(\xi_t^+,\xi_t^-,\xi_t^0, \hat{\xi}_t^0)\in E^+_{\rho_t}\oplus E^-_{\rho_t}\oplus E^0_{\rho_t}\oplus \hat{E}^0_{\rho_t}$, our {\color{black}first} aim is to show that $\frac{\|\xi_t^-\|^2}{\|\xi_t\|^2}$ converges to zero as $t\longrightarrow \infty$ uniformly in $\rho\in \Lambda$, $\xi\in T_\rho \Lambda$.

Thanks to (\ref{eq:DefHyp}) and (\ref{eq:EquivNorm}), we have 
$$\|\xi_t^-\|^2 \leq C A^{-2t} \|\xi^-\|^2 \leq \frac{C}{c} A^{-2t} \|\xi\|^2.$$

On the other hand, (\ref{eq:DefHyp}) and (\ref{eq:EquivNorm}) also imply that
$$\|\xi_t\|^2 \geq \frac{c}{C}  (|\xi^0|^2+ |\hat{\xi}^0|^2 +  \|\xi^+\|^2) \geq \frac{c}{C} |\sin \eta_0|^2 \|\xi\|^2.$$

Thus, we have
\begin{equation}\label{eq:AlejandroTuDoisCandidater}
\frac{\|\xi_t^-\|^2}{\|\xi_t\|^2}\leq \frac{C^2}{c^2} \frac{A^{-2t}}{|\sin \eta_0|^2},
\end{equation}
and the first claim follows.

We now move to the construction of $\tilde{E}_\rho^0$. Let $\rho\in \Phi^t(\Lambda)$ for $t\geq T_1$.

 The space $T_\rho \Phi^t(\Lambda)$ is $d$-dimensional, so {\color{black}its intersection with $E^-_{\rho}\oplus E^0_{\rho}\oplus \hat{E}^0_{\rho}$ must contain a} \textcolor{black}{norm-one} vector $\zeta=(0,\zeta^-, \zeta^0, \hat{\zeta}^0)\in E^+_{\rho}\oplus E^-_{\rho}\oplus E^0_{\rho}\oplus \hat{E}^0_{\rho}$. We denote by $\tilde{E}_\rho^0$ the vector space generated by $(0, 0,  \zeta^0, \hat{\zeta}^0)$. In particular, it depends \textcolor{black}{Hölder-}continuously on $\rho$, as claimed. \textcolor{black}{Let us also denote by $\boldsymbol{\zeta}^-$ the vector $(0,\zeta^-, 0, 0)$ and by $\boldsymbol{\zeta}^0$ the vector $(0, 0,  \zeta^0, \hat{\zeta}^0)$.}

Thanks to (\ref{eq:AlejandroTuDoisCandidater}), we know that\textcolor{black}{, for $t$ large enough, we have $|\zeta^-|< \frac{c_1 \delta_0}{6C_0}$, with $C_0$ as in (\ref{eq:CestRelouDeReprendreLePapier}) and $c_1$ as in (\ref{eq:CestRelouDeReprendreLePapier3}). Therefore, if $\delta_0$ is chosen small enough,  either $|\zeta^0|$ or $|\hat{\zeta}^0|$ must be $> \frac{1}{3}$. Up to exchanging the role of $\zeta^0$ and $\hat{\zeta}^0$, we may suppose that $|\zeta^0| \geq \frac{1}{3}$.}

\textcolor{black}{ Let $\xi= (\xi^+, \xi^-, \xi^0, \hat{\xi}^0)\in T_\rho \Phi^t(\Lambda)$. We denote by $\boldsymbol{\xi}^0$ the vector $(0, 0,  \xi^0, \hat{\xi}^0)$.
 If $\sigma$ denotes the natural symplectic structure on $T^*X$, we have thanks to (\ref{eq:CestRelouDeReprendreLePapier2}) that $\sigma(\xi, \zeta)= \sigma(\boldsymbol{\xi}^0, \boldsymbol{\zeta}^0) + \sigma(\xi, \boldsymbol{\zeta}^-)$. Recall from the discussion before (\ref{eq:CestRelouDeReprendreLePapier3}) that $\sigma(\boldsymbol{\xi}^0, \boldsymbol{\zeta}^0)= c(\rho) \left( \xi^0 \hat{\zeta}^0 -  \zeta^0 \hat{\xi}^0\right)$.}
 
\textcolor{black}{On the other hand, since $\Phi^t(\Lambda)$ is Lagrangian, $\sigma(\xi, \zeta)$ must be zero. Therefore, we have}
\begin{equation}
\textcolor{black}{\hat{\xi}^0 = \frac{ \xi^0 \hat{\zeta}^0}{\zeta^0} + \frac{1}{c(\rho) \zeta_0} \sigma(\xi, \boldsymbol{\zeta}^-),}
\end{equation}
\textcolor{black}{and the last term has an absolute value smaller than $\frac{\delta_0}{2} \|\xi\|$.} 

\textcolor{black}{Therefore, we may write}
$$
\textcolor{black}{\xi = (\xi^+, 0, 0, 0) + \frac{\xi^0}{\zeta^0} (0, 0, \zeta^0, \hat{\zeta}^0) + (0,0,0, \frac{1}{c(\rho)\zeta_0} \sigma(\xi, \boldsymbol{\zeta}^-)) +  (0, \xi^-, 0,0).}$$
\textcolor{black}{The result follows, as the first two terms belong to $E_\rho^+\oplus \tilde{E}_\rho^0$, while, for $t\geq T^1$, the sum of the last terms has absolute value smaller than $\delta_0 \|\xi\|$ thanks to (\ref{eq:AlejandroTuDoisCandidater}).}
\end{proof}

We may now proceed with the proof of Lemma \ref{lem:LaBorneChiante}. Recall that it says that,  if we consider the family of functions $D_t\in C^\infty(X)$ indexed by $t\geq T_0(\Lambda)$, defined by $D_t(x)=\log(|\det(d_x g_{0,t})|)$, then for all $\varepsilon>0$, there exists $\mu>0$ such that for all $t\geq T_0(\Lambda)$ and all $x,x'\in \Omega_{T_0}$ at mutual distance at most $\mu$,
\begin{equation}\label{eq:unif_cont}
|D_t(x)-D_t(x')|\leq \varepsilon\, .
\end{equation}

\begin{proof}[Proof of Lemma \ref{lem:LaBorneChiante}]
By compactness, we may find $\eta_0>0$ such that $\Lambda$ is $\eta_0$-transverse to the stable directions. Let $\varepsilon>0$, let $\delta_0>0$ which we will choose later, depending on $\varepsilon$, and let $T_1=T_1(\eta_0,\delta_0)$ be as in Lemma \ref{l:charts}, which we may assume to be greater than $T_0$. Clearly, it is enough to establish \eqref{eq:unif_cont} for $t\geq T_1$. Let $t\geq T_1$. By (\ref{eq:Fonctorialite}), for any $x\in\Omega_t$, defining $\rho\in\Lambda_t$ by $\pi(\rho)=x$, we have
\begin{equation}\label{eq:log_1}
\begin{aligned}
\mathrm{d}_xg_{0,t}&=\left(\mathrm{d}_{g_{T_1,t}(x)}g_{0,T_1}\right) \mathrm{d}_xg_{T_1,t}\\
&=(\mathrm{d}_{g_{T_1,t}(x)}g_{0,T_1})  (\mathrm{d}_{\Phi^{T_1-t}(\rho)}\pi|_{\Lambda_{T_1}})\br{\mathrm{d}_\rho \Phi^{T_1-t}|_{T_\rho\Lambda_t\to T_{\Phi^{T_1-t}(\rho)}\Lambda_{T_1}}} (\mathrm{d}_\rho\pi|_{\Lambda_t})^{-1}.
\end{aligned}
\end{equation}
To study the right-hand side of this decomposition, we make the three following observations.
\begin{itemize}
\item Thanks to the second part of Lemma \ref{l:charts}, for each $t\geq T_1$, we may find a subspace $\tilde{E}^0_\rho \subset E_\rho^0\oplus \hat{E}_\rho^0$ such that the orthogonal projector $P_\rho:T_\rho S^*X\rightarrow E_\rho^+\oplus \tilde{E}_\rho^0$ is H\"older continuous in $\rho$ and is $\delta_0$-close to the orthogonal projection onto $T_\rho\Lambda_s$. Therefore, up to an error $o_{\delta_0\to 0}(1)$, we may replace the factors $\mathrm{d}_{\rho'}\pi|_{\Lambda_s}$ in \eqref{eq:log_1} by factors $(\mathrm{d}_{\rho'}\pi P_{\rho'}^*)$, and the factor $\mathrm{d}_\rho \Phi^{T_1-t}|_{T_\rho\Lambda_t\to T_{\Phi^{T_1-t}(\rho)}\Lambda_{T_0}}$ by $P_{\Phi^{T_1-t}(\rho)}\mathrm{d}_\rho\Phi^{T_1-t} P_\rho^*$.

\item Thanks to (\ref{eq:Matrix}), we have $\det(P_{\Phi^{T_1-t}(\rho)}\mathrm{d}_\rho\Phi^{T_1-t} P_\rho^*) = \det( M_{\Phi^{T_1-t}(\rho), \rho})$.

\item For all $s\geq 0$, $\pi|_{\Lambda_s}:\Lambda_s\rightarrow \Omega_s$ is a diffeomorphism and uniformly bi-Lipschitz in $s$. Moreover, for each $s>0$, the maps $\Phi^{-s}:\Lambda_s\rightarrow\Lambda$ and $g_{s,s'}$ for $s<s'$ are contracting. Consequently, the mappings $x\mapsto \mathrm{d}_\rho\pi P_\rho^*$, $x\mapsto \mathrm{d}_{\Phi^{T_1-t}(\rho)}\pi P_{\Phi^{T_1-t}}^*$ and $x\mapsto \mathrm{d}_{g_{T_1,t}(x)}g_{0,T_1}$ (where $\pi(\rho)=x$) are H\"older continuous on $\Omega_t$, uniformly in $t\geq T_1$.
\end{itemize}

From these observations we deduce the existence of $\delta_0=\delta_0(\varepsilon)>0$ and $\mu_1=\mu_1(\varepsilon)>0$ such that if $T_1=T_1(\eta_0,\delta_0)>0$ is chosen accordingly, for any $t\geq T_1$ and $x,x'\in\Omega_t$ such that $\textup{dist}(x,x')\leq\mu_1$,
\begin{equation}\label{eq:log_3}
|D_t(x)-D_t(x')|\leq \varepsilon/2+|\log \det( M_{\Phi^{T_1-t}(\rho), T_1-t}) - \log \det( M_{\Phi^{T_1-t}(\rho'), T_1-t}) |,
\end{equation}
where $\rho,\rho'\in \Phi^t(\Lambda)$ are such that $\pi_X(\rho)= x$, $\pi_X(\rho')=x'$.

Now, by the chain rule, we have 
$$\det( M_{\Phi^{T_1-t}(\rho), \rho}) = \left(\prod_{k=0}^{\lfloor t- T_1 \rfloor- 1} \det( M_{\Phi^{T_1-t+k}(\rho), 1})\right) \times \det( M_{\Phi^{ - \langle t- T_1\rangle }(\rho), \langle t- T_1\rangle }),$$
where $\langle s\rangle = s- \lfloor s \rfloor$. We thus get that
\begin{align*}
&|\log \det( M_{\Phi^{T_1-t}(\rho), T_1-t}) - \log \det( M_{\Phi^{T_1-t}(\rho'), T_1-t}) | \\
&\leq \sum_{k=0}^{\lfloor t- T_1 \rfloor- 1} \left|\log \det( M_{\Phi^{T_1-t+k}(\rho), 1}) - \log \det( M_{\Phi^{T_1-t+k}(\rho{\color{black}'}), 1})\right| \\
&+ \left|\log \det( M_{\Phi^{ - \langle t- T_1\rangle }(\rho), \langle t- T_1\rangle }) - \det( M_{\Phi^{ - \langle t- T_1\rangle }(\rho'), \langle t- T_1\rangle })\right|.
\end{align*}

We note that the map $\rho'\mapsto \log \det( M_{\Phi^{ - s }(\rho'), s })$ is Hölder continuous uniformly in $s\in[0,1]$, and that the distance between $\Phi^{T_1-t+k}(\rho)$ and $\Phi^{T_1-t+k}(\rho')$ decays exponentially in $k$ as $|T_1-t+k|$ increases from $1$ to $T_1-t$, with an exponent of decay independent of $\rho, \rho'$. We deduce from this that 
$$|\log \det( M_{\Phi^{T_1-t}(\rho), T_1-t}) - \log \det( M_{\Phi^{T_1-t}(\rho'), T_1-t}) | \leq C \text{dist}_{S^*X}(\rho, \rho')^\alpha,$$
for some $C,\alpha$ independent of $\rho, \rho'$. The result follows from this and (\ref{eq:log_3}).
\end{proof}

\subsection{The action of the Schrödinger propagator on Lagrangian states}\label{ss:demo_propagation}
The aim of this section is to prove Proposition \ref{prop:SumLag3}, which describes the action of the Schrödinger propagator on Lagrangian states that are transverse to the stable directions. We will start with the following proposition, which treats the case of Lagrangian states that are nowhere stable, on a complete simply connected manifold. The discussion is simplified by the fact that, here, we only consider Lagrangian states associated with Lagrangian manifolds that are projectable. We start by establishing the following proposition, which is an adaptation of results from \cite[Section 4.1]{NZ}.
\begin{tcolorbox}[breakable, enhanced]
\begin{proposition}\label{cor:SchrodHadam}
Let $\widetilde{X}$ be the universal cover of a Riemannian manifold of negative sectional curvature, let $0<\lambda_1<\lambda_2$, and
let $\Lambda \subset S^*_{[\lambda_1,\lambda_2]}$ be a projectable Lagrangian submanifold of $T^*\widetilde{X}$ with \textcolor{black}{support} $\Omega_0$, generated by a phase function $\phi_0$. Suppose that $\Lambda$ is nowhere stable. Then there exists $T_0\geq 0$ such that, for any $t\geq T_0$,  $\widetilde{\Phi}^t(\Lambda)$ is a projectable Lagrangian manifold, whose \textcolor{black}{support} and phase we denote by $\Omega_t$ and $\phi_t$. Furthermore, for any symbol $a\in C^\infty_c(\Omega_0)$, any $t\geq T_0$ and any $k\in \N$, the application of the operator $U_h(t)$ to the  Lagrangian state
\begin{equation*}
a e^{i\phi_0/h}
\end{equation*}
associated with $\Lambda_0$ can be written as
\begin{equation*}
U_h(t) (a e^{i\phi_0/h})(x) = e^{i\phi_{t}(x)/h+i\beta(t)/h} b_{t}(x) +O_{H^k}(h)
\end{equation*}
for some $\beta(t)\in \R$. Here, the $b_{t}$ are smooth compactly supported functions on $\widetilde{X}$ such that $\phi_{t}$ is smooth on a neighbourhood of the support of $b_{t}$. The two functions are defined as follows:
\begin{enumerate}
\item Let $\pi_{\widetilde{X}}:T^*\widetilde{X}\rightarrow\widetilde{X}$ be the canonical projection. For any $y\in \Omega_t$, there exists a unique $x\in \Omega_0$ such that $\pi_{\widetilde{X}}\widetilde{\Phi}^t(x, \partial \phi_0(x)) =y$. We then write  $\widetilde{\Phi}^t(x, \partial \phi_0(x)) = (y, \partial \phi_t(y))$.
\item The function $b_t$ is defined by
\[
b_t(x)=|\det(dg_t(x))|^{1/2}a\circ g_t(x)
\]
where $g_t(x)= g_{0,t}(x)$ is the projection on the base manifold of $\widetilde{\Phi}^{-t}(x,\partial_t\phi(x))$. Furthermore, $\log|b_t(x)|$ is continuous in $x$, uniformly in $(x,t)$.
\end{enumerate}
\end{proposition}
\end{tcolorbox}

\begin{proof}
This proposition essentially follows from the fact that the Schrödinger propagator is a Fourier Integral Operator, and by using the WKB method. This method has been described in \cite[Lemma 4.1]{NZ}, in coordinate charts. Let us explain how we can reduce the proof to this setting. In the coming steps we will use tools from semiclassical analysis, some of which are presented in section \ref{sec:semi-classique} of the appendix.\\

\noindent\textbf{Step 1: the Schrödinger propagator as a Fourier Integral Operator}\\
For any $t\in \R$, we denote
\[
\Lambda(t):= \{(x,x'; \xi, -\xi'); \Phi^t(x,\xi)= (x', \xi')\} \subset T^*(\widetilde{X}^2)\, .
\]
We claim that {\color{black} if $t\neq 0$,} $\Lambda(t)$ is a projectable Lagragian submanifold of $T^*(\widetilde{X}^2)$. Indeed, since $\widetilde{X}$ has negative curvature,  \cite[Theorem 4.8.1]{Jost} implies that for any $x,x'\in \widetilde{X}$ and any $t\in \R{\color{black}\setminus\{0\}}$, there exists unique $\xi\in T_x\widetilde{X}$, $\xi'\in T_{x'}\widetilde{X}$, depending smoothly on $x$ and $x'$, such that $\Phi^t(x,\xi)= (x', \xi')$. In other words, $\Lambda(t)$ is a smooth section of $T^*(\widetilde{X}^2)$, so it is a projectable Lagrangian manifold thanks to Remark \ref{Rem:CriterionProj}. Next, recall the standard fact that the {\color{black} frequency-localized} Schrödinger propagator is a Fourier Integral Operator associated to the geodesic flow, whose proof is similar to  \cite[Theorem 2.1]{LapSig} (see also  \cite[Theorem 10.4]{Zworski_2012}). This means that, if $\psi_t$ is a phase generating $\Lambda(t)$, and if\footnote{The space $\Psi_h^{comp}$ is the space of (Weyl) pseudo-differential operators with \textcolor{black}{compactly supported} symbols. It is defined in section \ref{sec:semi-classique} of the appendix.} $A\in \Psi_h^{comp}(\widetilde{X})$, then there exists \textcolor{black}{$u\in C_c^\infty(\widetilde{X}^2)$} such that $U(t)A$ is the sum of an operator whose Schwartz kernel  is 
\[
u(x,x') e^{\frac{i}{h} \psi_t(x,x')}
\]
and of an $O_{L^2\longrightarrow L^2}(h^\infty)$ remainder. \textcolor{black}{Here, $u$ can depend on $h$, but its supports and $C^k$ norms are bounded independently of $h$}.\\

\noindent\textbf{Step 2: Using coordinates}\\
Fourier Integral Operators are easier to describe in some system of coordinates. Since $\widetilde{X}$ is a complete simply connected manifold of negative curvature, by the Cartan-Hadamard theorem, there exists a diffeomorphism $\kappa : \widetilde{X} \longrightarrow \R^d$, which is simply given by the exponential map at any point. We denote by $T^*\widetilde{X}$ the co-tangent bundle of $\widetilde{X}$, and by $\Phi^t : T^*\widetilde{X}\longrightarrow T^* \widetilde{X}$ the geodesic flow at time $t$. We equip $T^*\widetilde{X}$ with its natural symplectic structure. The diffeomorphism $\kappa$ can be lifted to a \emph{symplectomorphism} by
\begin{equation}\label{eq:SymplLift}
K :
\begin{cases}
T^* \widetilde{X} \longrightarrow T^* \R^d \\
(y,\xi)\mapsto \big{(}\kappa(y), (d\kappa(y)^t)^{-1}\xi\big{)}.
\end{cases}
\end{equation}

For any $t\in \R$, let us write
\[
\widehat{\Lambda}(t):= \left\{\big{(}\kappa(x), \kappa(x') ; (d\kappa(x)^t)^{-1}\xi, (d\kappa(x')^t)^{-1}\xi'\big{)} ~|~  (x,x',\xi,\xi')\in \Lambda(t) \right\}\subset T^*\R^{2d}\, .
\]
We deduce from the previous step that $\widehat{\Lambda}(t)$ is a projectable Lagrangian submanifold of $T^*\R^{2d}$. 
Furthermore, if $\widehat{\psi}_t$ is a phase generating $\widehat{\Lambda}(t)$,
 and if $A\in \Psi_h^{comp}(\widetilde{X})$, then there exists \textcolor{black}{$\widehat{u}\in C_c^\infty(\R^d)$} such that 
$\tilde{U}_h(t):= (\kappa^{-1})^* U_h(t)A (\kappa)^* : L^2(\R^d)\longrightarrow L^2(\R^d)$
is the sum of an operator whose Schwartz kernel  is 
\begin{equation}\label{eq:IntKer}
\widehat{u}(y,y') e^{\frac{i}{h} \widehat{\psi}_t(y,y')},
\end{equation}
and of an $O_{L^2\longrightarrow L^2}(h^\infty)$ remainder. \textcolor{black}{Here, again, $\widehat{u}$ can depend on $h$, but its supports and $C^k$ norms are bounded independently of $h$}.\\

\noindent\textbf{Step 3: Using the WKB method}\\
Let $\phi_0\in C^\infty(\Omega_0)$ and $a\in C_c^\infty(\Omega_0)$ be as in the statement of Proposition \ref{cor:SchrodHadam}. Recall that $\lambda_1\leq |\partial\phi_0|\leq \lambda_2$.
Take $\alpha_1\in C_c^\infty(\tilde{X})$ with $\alpha_1\equiv 1$ on the support of $a$, and $\alpha_2\in C_c^\infty(0, +\infty)$ with $\alpha_2\equiv 1$ on $[\lambda_1, \lambda_2]$. Then, if we write $\alpha : T^*X\ni (x,\xi) \mapsto \alpha_1(x) \alpha_2(|\xi|)\in \R$ and $A:= \mathrm{Op}_h (\alpha)$, the method of stationnary phase shows that $a e^{\frac{i}{h}\phi_0} = A a e^{\frac{i}{h}\phi_0} + O_{H^k}(h^\infty)$ for any $k\in\N$. Let us write $\omega_0:= \kappa(\Omega_0)\subset \R^d$, $\varphi_0 := (\kappa^{-1})^*\phi_0\in C^\infty(\omega_0)$, and $\mathrm{a} := (\kappa^{-1})^*a\in C_c^\infty(\omega_0)$. Moreover, for each $t\in\R$, let us write $\widehat{\Phi}^t=K\circ\Phi^t\circ K^{-1}:T^*\R^d\rightarrow T^*\R^d$.\\

We thus want to apply the operator $(\kappa^{-1})^* U_h(t)(\kappa)^*$ to the Lagrangian state $\mathrm{a} e^{\frac{i}{h}\varphi_0}$. Up to a $O_{H^k}(h^\infty)$ (for any $k\in\N$), it does therefore amount to applying the operator $ (\kappa^{-1})^* U_h(t)A (\kappa)^*$, whose integral kernel is described by (\ref{eq:IntKer}).
We are then exactly in the framework of \cite[Lemma 4.1]{NZ}, and we can conclude using the following lemma, the last point in Proposition \ref{cor:SchrodHadam} coming from Lemma \ref{lem:LaBorneChiante}.
 
\begin{tcolorbox}[breakable, enhanced]
\begin{lemma}\label{lem:SchrodEucl} 
Let $\Lambda_0$ be a projectable Lagrangian submanifold of $T^*\R^d$, generated by a phase function $\varphi_0$, with \textcolor{black}{support} $\omega_0\subset \R^d$. Fix $t>0$. Suppose that $\widehat{\Phi}^t(\Lambda_0)$ is a projectable Lagrangian submanifold of $T^*\omega_t$, for some open subset $\omega_t\subset\R^d$, generated by a phase function $\varphi_t$. Then, for any $\mathrm{a}\in C_c^\infty(\omega_0)$ and any $k\in \N$, the application of the operator $(\kappa^{-1})^* U_h(t)(\kappa)^*$  to the  Lagrangian state
\begin{equation*}
\mathrm{a} e^{i\varphi_0/h}
\end{equation*}
associated with $\Lambda_0$ can be written as
\begin{equation*}
(\kappa^{-1})^* U_h(t)(\kappa)^*(\mathrm{a} e^{i\varphi_0/h})(x) = e^{i\varphi_t/h} \left(\mathrm{a}_t+ r_h\right),
\end{equation*}
where $\|r_h\|_{H^k}= O(h)$, and where $\mathrm{a}_t\in C_c^\infty(\R^d)$ is given by
\begin{equation}\label{eq:FormulaAmplitude}
\mathrm{a}_t(x)= e^{\frac{i}{h}\beta(t)} |\det \mathrm{d}\mathrm{g}_{0,t}(x)|^{\frac{1}{2}} (\mathrm{a}\circ \mathrm{g}_{0,t})(x),
\end{equation}
for some $\beta(t)\in \R$. Here, $\mathrm{g}_{0,t} : \omega_t\longrightarrow \omega_0$ is given by
$\mathrm{g}_{0,t} = \kappa^{-1} \circ g_{0,t} \circ \kappa$, with $\mathrm{g}_{0,t} :\omega_t\longrightarrow \omega_0$
as in (\ref{eq:DefG}).
\end{lemma}
\end{tcolorbox}
\end{proof}

We may now proceed with the proof of Proposition \ref{prop:SumLag3}, after introducing a few notations. Let us write $\pr : \widetilde{X}\longrightarrow X$ for the covering map of $X$. It induces a projection $\pr^* : T^*\widetilde{X}\longrightarrow T^*X$, such that $\pr^*\circ \tilde{\Phi}^t = \Phi^t \circ \pr^*$. We shall write $\pr^{-1}(x) := \{y\in \widetilde{X}~|~ \pr (y) = x\}$. We also define a map $\Pi: C^0(\tilde{X}) \longrightarrow C^0(X)$ by $(\Pi f)(x) = \sum_{y\in \pr^{-1}(x)} f(y)$.
Let us denote by $\tilde{U}_h(t) : L^2(\tilde{X})\longrightarrow  L^2(\tilde{X})$ the semi-classical Schrödinger propagator on $\tilde{X}$. If $f\in C^0(\tilde{X})$, we have
\begin{equation}\label{eq:PropagRevet}
\Pi \tilde{U}_h(t) f = U_h(t) \Pi f,
\end{equation}
because both side satisfy the same differential equation with the same initial conditions.

\begin{proof}[Proof of Proposition \ref{prop:SumLag3}]
Thanks to Lemma \ref{lem:PetitesLag}, we know that we may find finitely many open sets $(\mathcal{O}_n)_{n=1,\dots,N}$ in $X$ such that $T^*\mathcal{O}_n\cap \Lambda$ is nowhere stable. Let $(\chi_n)_{n=1,\dots,N}$ be a family of smooth functions with $\chi_n\in C_c^\infty(\mathcal{O}_n)$ and $\sum_{n=1}^N \chi_n \equiv 1$ on $\Omega_0$. For each $n\in \{1,\dots,N\}$, let $\Lambda_n = \Lambda\cap T^*\mathcal{O}_n\subset T^*X$, which is a projectable nowhere stable Lagrangian submanifold with support $\mathcal{O}_n$. Then, there exists a projectable Lagrangian submanifold $\widetilde{\Lambda}_n\subset T^*\tilde{X}$ such that the projection $\textup{pr}^*$ restricts to a diffeomorphism $\widetilde{\Lambda}_n\rightarrow\Lambda_n$. For the rest of the proof $T_0=\max_n T_0(\Lambda_n)\vee T_0(\tilde{\Lambda}_n)$ as in Lemma \ref{lem:LongTimeProj}. Moreover, any other projectable Lagrangian $\widetilde{\Lambda}'_n\subset T^*\tilde{X}$ with the same property is the image of $\widetilde{\Lambda}_n$ by some element of $\Gamma$. We call $\widetilde{\Lambda}_n$ a \emph{lift} of $\Lambda_n$. Note that $\Lambda_n\subset S_{[\lambda_1,\lambda_2]}^*X$ if and only if $\widetilde{\Lambda}_n\subset S_{[\lambda_1,\lambda_2]}^*\tilde{X}$. Let us fix a lift of each $\Lambda_n$, and denote by $\widetilde{\Omega}_n$ and $\widetilde{\phi}_n$ its \textcolor{black}{support} and phase function. We may also lift each $a\chi_n\in C_c^\infty(\mathcal{O}_n)$ to some $\widetilde{a\chi_n}\in C_c^\infty(\widetilde{\mathcal{O}_n})$ such that $a\chi_n= \widetilde{a\chi_n}\circ \pr$. Hence, $\Pi \left(\widetilde{a\chi_n} e^{\frac{i}{h}\widetilde{\phi}_0}\right) = a\chi_n e^{\frac{i}{h}\phi_0}$. We then apply Corollary \ref{cor:SchrodHadam} to describe the action of $\tilde{U}_h(t)$ on the Lagrangian state $\widetilde{a\chi_n}e^{\frac{i}{h}\widetilde{\phi_0}}$. Let $b_{t,n}$ and $\phi_{t,n}$ be the phases appearing in the statement of the proposition. Thanks to equation (\ref{eq:PropagRevet}), we get

\begin{align*}
U_h(t) \left(a e^{i\phi_0/h}\right) = \sum_{n=1}^N U_h(t) \left(a\chi_n e^{i\phi_0/h}\right) &= \sum_{n=1}^N \Pi \Big{(}b_{t,n} e^{\frac{i}{h}\phi_{t,n}}\Big{)} +O_{C^0}(h).
\end{align*}

Now, the group of deck transformations of $X$ is freely acting, properly discontinuous group of isometries of $\tilde{X}$, and  for each $t\geq T_0$, $b_{t,n}$ has compact support. Therefore, $\Pi \Big{(}b_{t,n} e^{\frac{i}{h}\phi_{t,n}}\Big{)}(x) = \sum_{y\in \pr^{-1}(x)} b_{t,n} e^{\frac{i}{h}\phi_{t,n}} $ is made of finitely many terms. Equation (\ref{eq:PropagLagSum2}) follows. The rest of the statements follow from the corresponding properties in Corollary \ref{cor:SchrodHadam}, noting that (\ref{eq:LaBorneLog}) is equivalent to the fact that $\log |b_{t,n}|$ is continuous, uniformly in $(x,t)$.\\

It remains to prove equation \eqref{eq:phase_bound}. Suppose for contradiction that this bound does not hold. Then, we could find a sequence $t_p$ of times larger than $T_0$, a sequence $j_p\in \{1,\dots, M(t_p)\}$ and sequences of points $x_p, y_p\in \mathcal{U}'_{j_p, t_p}$ such that
\[
\textup{dist}_{X}(x_p,y_p)=o(1)\textup{ and }\textup{dist}_{X}(x_p,y_p) = o\left(\textup{dist}_{T^*X} \left((x_p,\partial \phi_{j_p,t_p}(x_p)), (y_p , \partial \phi_{j_p,t_p}(y_p))\right)\right)\, ,
\]
where the distance on $T^*X$ is computed thanks to the metric $g_0$ we fixed. Now, the points $(x_p, \partial \phi_{j_p,t_p}(x_p))$ and $(y_p, \partial \phi_{j_p,t_p}(y_p))$ can be lifted to points $(\widetilde{x_p}, \widetilde{\xi_p})$ and $(\widetilde{y_p}, \widetilde{\xi'_p})$ in $\Phi^{t_p}(\widetilde{\Lambda}_n)\subset T^*\widetilde{X}$. But we would then have $\textup{dist}_{\widetilde{X}}(\widetilde{x_p},\widetilde{y_p})=o(1)$ and
 $\textup{dist}_{\widetilde{X}}(\widetilde{x_p},\widetilde{y_p})=o \left(d_{T^*\widetilde{X}} \left((\widetilde{x_p}, \widetilde{\xi_p}), (\widetilde{y_p}, \widetilde{\xi'_p})\right)\right)$. In particular, for $p$ large enough, the map $t\mapsto \textup{dist}_{\widetilde{X}}\left(\widetilde{\Phi}^{t}((\widetilde{x_p}, \widetilde{\xi_p}),(\widetilde{y_p}, \widetilde{\xi'_p})\right)$ would not be increasing for $t$ close to zero,  which would contradict the fact that $\widetilde{\Lambda}_n$ is expanding.

\end{proof}
\section{Rational independence of phases is generic}\label{sec:PhaseGen}

Let $\Omega\subset X$ be an open subset. Let $0<\lambda_1<\lambda_2$. Let $\Sigma\subset X$ be an orientable hypersurface of $X$. We fix $\nu$ a vector field on $\Sigma$ normal to $T\Sigma$ at each point and of unit norm. Recall the objects $\mathcal{E}_{(\lambda_1,\lambda_2)}(\Omega)$, $\mathcal{E}^T_{(\lambda_1,\lambda_2)}(\Omega)$ and $\mathcal{E}^{T,irr}_{(\lambda_1,\lambda_2)}(\Omega)$ defined in \eqref{eq:SpacePhases}, \eqref{eq:TransversePhases} and \eqref{eq:DefIrr} respectively. Moreover, recall $\calC(\Sigma)$, $\calC^T(\Sigma)$ and $\mathcal{C}^{T,irr}(\Sigma)$ defined in \eqref{eq:calC_def}, \eqref{eq:calC_trans} and \eqref{eq:DefGeneMono} respectively. The goal of the present section is to prove the two following propositions, which we use in the proofs of Theorems \ref{th:Pointwise} and \ref{theoMono} respectively. We equip all these sets with the topology of uniform convergence of derivatives on compact sets.

\begin{tcolorbox}[breakable, enhanced]
\begin{proposition}\label{prop:polychrom_conclusion}
The set $\mathcal{E}^{T,irr}_{(\lambda_1,\lambda_2)}(\Omega)$ is a residual subset of $\mathcal{E}^{T}_{(\lambda_1,\lambda_2)}(\Omega)$.
\end{proposition}
\end{tcolorbox}

\begin{tcolorbox}[breakable, enhanced]
\begin{proposition}\label{prop:monochrom_conclusion}
The set $\mathcal{C}^{T,irr}(\Sigma)$ is a residual subset of $\mathcal{C}^{T}(\Sigma)$.
\end{proposition}
\end{tcolorbox}

\subsection{The polychromatic case: proof of Proposition \ref{prop:polychrom_conclusion}}\label{subsec:poly}

In this subsection we prove Proposition \ref{prop:polychrom_conclusion}. For the proof, we will need the following definition. Let $k$ be an integer no smaller than two. For each finite sequence of relative, non-zero integers, $\boldsymbol{n}=(n_1,\dots,n_k)$, let
\begin{align*}
\calT_{\boldsymbol{n}}=\Big{\{}((x,\xi_1),&\dots,(x,\xi_k))\, :\, x\in X,\, \xi_1,\dots,\xi_k\in T^*_xX\setminus\{0\},\\
&\text{such that the $\xi_j$ are not all equal and }  \sum_{j=1}^kn_j\xi_j=0\Big{\}}\, .
\end{align*}
Proposition \ref{prop:polychrom_conclusion} will be a consequence of the following result.
\begin{tcolorbox}[breakable, enhanced]
\begin{lemma}[Rationally independent polychromatic phases are generic]\label{lem:polychrom_irr}
Let $k\geq 2$, $\boldsymbol{n}\in \N^k$.
There exists a residual subset $E_{\boldsymbol{n}}\subset C^\infty(\Omega)$ such that for each $\phi\in E_{\boldsymbol{n}}$,
the set 
$$\mathcal{O}_{\boldsymbol{n}} := \left\{ (x_1,\dots,x_k,t) \in \Omega^k\times \R \text{ such that } (\Phi^{t}(x_1,\partial \phi(x_1)), \dots , \Phi^{t}(x_k,\partial \phi(x_k)))\in\calT_{\boldsymbol{n}} \right\}$$
is a countable union of one dimensional submanifolds.
\end{lemma}
\end{tcolorbox}
Before proving Lemma \ref{lem:polychrom_irr} let us deduce Proposition \ref{prop:polychrom_conclusion} from it.
\begin{proof}[Proof of Proposition \ref{prop:polychrom_conclusion}]
A countable intersection of residual sets is still a residual set. Hence, thanks to the previous lemma, we know that there exists a residual subset $E\subset C^\infty(\Omega)$ such that for all $\phi\in E$, the following holds. For all $k\geq 2$ and all $\boldsymbol{n}\in \N^k$, the sets $\mathcal{O}_{\boldsymbol{n}}$ are countable unions of one dimensional submanifolds. Let $\phi \in E$, and let $k\geq 2$. We shall write $\Psi_\phi : \Omega^k \times \R \ni (x_1,\dots, x_k,t) \mapsto \pi_X \left(\Phi^t(x_1, \partial \phi(x_1) )\right)\in X$. Then, for all $\boldsymbol{n}\in \N^k$, $\Psi_\phi (\mathcal{O}_{\boldsymbol{n}})$ has measure zero. 

{\color{black}We claim that} the set  $\Psi_\phi (\mathcal{O}_{\boldsymbol{n}})$ is exactly the set of $x\in X$ such that there exists $x_1,\dots, x_k\in X$, $t\in\R$ and $\xi_1,\dots, \xi_k \in T^*_x X$ such that $\Phi^t(x_j, \partial \phi(x_j)) = (x,\xi_j)$ for all $j=1,\dots, k$ and $\sum_{j=1}^k n_j \xi_j= 0$. Indeed, by the discussion after Proposition \ref{prop:SumLag3}, the directions $\xi_1,\dots, \xi_k$ are all different{\color{black}, so the claim follows from definition of $\calT_{\boldsymbol{n}}$}. 

{\color{black}All in all}, if $\phi\in E\cap \mathcal{E}^T_{(\lambda_1,\lambda_2)}(\Omega)$, we have that for almost every $x_0\in X$, the vectors $(\xi_j^{t,x_0})_{j=1,\dots,N_{x_0}(t)}$ are rationally independent for all $t\geq T_0(\phi)$. This is precisely saying that $E\cap \mathcal{E}^T_{(\lambda_1,\lambda_2)}(\Omega)\subset \mathcal{E}^{T,irr}_{(\lambda_1,\lambda_2)}(\Omega)$, which proves the result.
\end{proof}

\begin{proof}[Proof of Lemma \ref{lem:polychrom_irr}]
Let us write
\begin{align*}
\Psi:(T^*\Omega)^k\times\R&\rightarrow (T^*X)^k\\
((x_1,\xi_1),\dots,(x_k,\xi_k),t)&\mapsto \left( \Phi^{t}(x_1,\xi_1),\dots,\Phi^{t}(x_k,\xi_k)\right)\, .
\end{align*}

Since $\Phi^{t}$ is a diffeomorphism, the map $\Psi$ is a submersion. Moreover, $\calT_{\boldsymbol{n}}$ is a submanifold of $(T^*X)^k$ of codimension $kd$. Therefore, $\Psi^{-1}(\calT_{\boldsymbol{n}})$ is a submanifold of $(T^*\Omega)^k\times\R$ of codimension $kd$ (and hence of dimension $kd+1$). Let $\textup{pr}_{(T^*\Omega)^k}$ denote the projection of $(T^*\Omega)^k\times\R$ onto $(T^*\Omega)^k$. We claim that

\begin{equation}\label{eq:claim}
\textup{dim}(\textup{Ker}(d\textup{pr}_{(T^*\Omega)^k)})\cap T\Psi^{-1}(\calT_{\boldsymbol{n}}))=0.
\end{equation}

The proof of (\ref{eq:claim}) is a geometric argument which we postpone to the end of the proof of the present lemma. By (\ref{eq:claim}), $P_{\boldsymbol{n}}=\textup{pr}_{(T^*\Omega)^k}(\Psi^{-1}(\calT_{\boldsymbol{n}}))$ is a countable union of submanifolds of $(T^*\Omega)^k$ of dimension $kd+1$. By the multijet transversality theorem (see Theorem 4.13, Chap.2 of \cite{GG}) the set $E_{\boldsymbol{n}}$ of phases $\phi\in\calE_{(\lambda_1,\lambda_2)}(\Omega)$ such that the section of $(T^*\Omega)^k$ defined by $(d\phi,\dots,d\phi)$ is transversal to $P_{\boldsymbol{n}}$ is a residual subset of $C^\infty(\Omega)$. In particular, if $\phi \in E$, the intersection of the section $(d\phi,\dots,d\phi)$ with $P_{\boldsymbol{n}}$ is a countable union of submanifolds of dimension 1 of $(T^*\Omega)^k$.
Using (\ref{eq:claim}) again, we deduce that $\mathcal{O}_{\boldsymbol{n}}$ is a countable union of one dimensional submanifolds of $\Omega^k\times \R$.\\

To conclude, we now prove \eqref{eq:claim}.
Firstly, for all $p\in (T^*\Omega)^k\times\R$, $\textup{Ker}(d_p\textup{pr}_{(T^*\Omega)^k})=\{0\}\times\R$. In order to prove the statement of the claim, we must therefore study the image by $d_p\Psi$ of perturbations of $p$ along the $t$ variable Let $\tau =(0,\dots,0, t)\in \{0\}\times\R^k\subset T_p((T^*\Omega)^k\times\R^k)$ be different from zero. We want to show that $\tau \notin T_p \Psi^{-1}(\mathcal{T}_{\boldsymbol{n}})$. Since $\Psi$ is a submersion, this amounts to proving that $d_p\Psi(\tau) \notin T_{\Psi(p)}\calT_{\boldsymbol{n}}$. Let us write $d_p\Psi(\tau)=((v_1,w_1),\dots,(v_k,w_k))\in (T_x^*X)^k$. Then, we have $v_j= t \sigma_x(\xi_j)$ for $j=1,\dots,k$, where $\sigma_x : T_x^*X\longrightarrow T_xX$ is the identification of the cotangent and tangent spaces induced by the Riemannian metric. 
But on the other hand, for $d_p\Psi(\tau)$ to belong to $T_{\Psi(p)}\calT_{\boldsymbol{n}}$, the vectors $v_k$ would have to be all equal.  Since the $\xi_j$ are not all equal, this would imply $\tau=0$. Hence,
\[
\textup{Ker}(d_p\textup{pr}_{(T^*\Omega)^k)})\cap T_p\Psi^{-1}(\calT_{\boldsymbol{n}})=\{0\}\times\R\cap d_p\Psi^{-1}(T_p\calT_{\boldsymbol{n}})=\{0\}
\]
as announced.
\end{proof}

\subsection{The monochromatic case: proof of Proposition \ref{prop:monochrom_conclusion}}\label{subsec:mono}

The aim of this section is to prove Proposition \ref{prop:monochrom_conclusion}. As in section \ref{subsec:poly}, we will use an intermediate result, Proposition \ref{prop:irrational_monochromatic} below, for which we now introduce certain geometric objects.\\

Let us denote by $\hat S^*\Sigma$ for the set of $(x,\xi)\in S^*X$ with $x\in \Sigma$ and $\xi\notin T_x^*\Sigma$. If $k\in \N$, consider the map
$$\Psi^k : \begin{cases}  \hat{S}^*\Sigma^k  \times \R^k  &\longrightarrow (S^*X)^k\\
((x_1,\xi_1),\dots, (x_k, \xi_k),t_1,\dots,t_k) &\mapsto \left(\Phi^{t_1}(x_1,\xi_1),\dots, \Phi^{t_k}(x_k,\xi_k)\right).
\end{cases}$$
The differential of $\Psi^k$ at any point is invertible, so, by the inverse function theorem, the image of $\Psi^k$ is an open set, which we will call the \emph{reachable set}, and denote by $\mathcal{R}^k(\Sigma)\subset (S^*X)^k$. For each finite sequence of non-zero relative integers $\boldsymbol{n}=(n_1,\dots,n_k)$ with $k\geq 2$ let
\[
\calS_{\boldsymbol{n}}=\left\{((x,\xi_1),\dots,(x,\xi_k))\st x\in X,\ \xi_1,\dots,\xi_k\in S^*_xX,\ \sum_{i=1}^k n_i\xi_i=0\right\}\subset (S^*X)^k\, .
\]

Let $\pi : (\hat{S}^*\Sigma)^k \times \R^k \longrightarrow (\hat{S}^*\Sigma)^k$ denote the projection on the first component.
We will write 
\begin{align*}
Z_{\boldsymbol{n}} &:= (\Psi^k)^{-1} \left(\calS_{\boldsymbol{n}} \cap \mathcal{R}^k(\Sigma) \right)\subset (\hat{S}^*\Sigma)^k \times \R^k \\
Z'_{\boldsymbol{n}} &:= \pi(Z_{\boldsymbol{n}})\\
Z''_{\boldsymbol{n}} &:=p_\Sigma^k (Z'_{\boldsymbol{n}}) \subset (T^*\Sigma)^k,
\end{align*}
where, if $(x,\xi)\in \hat{S}^*\Sigma$, we write $p_\Sigma(x,\xi) = (x,\zeta)\in T^*\Sigma$, where $\zeta$ is the orthogonal projection of $\xi$ on $T_x^*\Sigma$ and $p_\Sigma^k$ acts as $p_\Sigma$ on each coordinate of $\hat{S}^*\Sigma$. In particular, $|\zeta|<1$. Proposition \ref{prop:monochrom_conclusion} will follow from Proposition \ref{prop:irrational_monochromatic} below.

\begin{tcolorbox}[breakable, enhanced]
\begin{proposition}\label{prop:irrational_monochromatic}
Let $k\in \N \setminus \{1\}$ and $\boldsymbol{n}\in \N^k$. There exists a residual subset $E_{\boldsymbol{n}}\subset \calC(\Sigma)$ such that for all $u\in E_{\boldsymbol{n}}$, the set 
$$\left\{(x_1,\dots, x_k)\in \Sigma^k \text{ such that } ((x_1, \partial u(x_1)),\dots,(x_k, \partial u(x_k)) \in Z''_{\boldsymbol{n}} \right\}.$$
is empty.
\end{proposition}
\end{tcolorbox}

Before proving Proposition \ref{prop:irrational_monochromatic}, let us show that it implies Proposition \ref{prop:monochrom_conclusion}.
\begin{proof}[Proof of Proposition \ref{prop:monochrom_conclusion}]
By Proposition \ref{prop:irrational_monochromatic}, $E= \bigcap_{k\in \N\setminus\{1\}, \boldsymbol{n}\in \N^k} E_{\boldsymbol{n}}$ is a residual set. If $u\in E$, then for all $k\in \N\setminus\{1\}$, for all $x_1,\dots, x_k\in \Omega_u$ and all $t_1,\dots, t_k$, if there exists $x\in X$ such that for all $j=1,\dots,k$, we have $\Phi^{t_j}(x_j, \partial \phi_u(x_j)) = (x,\xi_j)$, then the $\xi_j$ are rationally independent. Indeed, by \eqref{eq:phi_from_u},
\[
((x_1,\partial u(x_1),\dots,(x_k,\partial \phi_u(x_k))\in Z''_{\boldsymbol{n}}\Leftrightarrow ((x_1,\partial u(x_1),\dots,(x_k,\partial \phi_u(x_k))\in Z'_{\boldsymbol{n}}
\]
Therefore, if $u\in E \cap \calC^T(\Sigma)$, we have $u\in \mathcal{C}^{T,irr}(\Sigma)$.
\end{proof}

\begin{proof}[Proof of Proposition \ref{prop:irrational_monochromatic}]
By the inverse function theorem, for any $q\in (\hat{S}^*\Sigma)^k \times \R^k$, there exists a neighbourhood $V_q$ of $\Psi^k(q)$ and a map  $\mathcal{I}^k_q : V_q \longrightarrow  (\hat{S}^*\Sigma)^k \times \R^k$ such that $\Psi^k \circ \mathcal{I}^k_q \equiv \mathrm{Id}$ on $V_q$, and $\mathcal{I}^k_q (\Psi^k(q)) = q$. We define the \emph{hitting map} by $H_q = \pi \circ \mathcal{I}^k_q : V_q \longrightarrow (S^*X)^k$. We would like to show that $Z'_{\boldsymbol{n}}$ is a countable union of submanifolds of $(S^*X|_\Sigma)^k$ of codimension at least $kd-1$. To this end, we note that $Z'_{\boldsymbol{n}}$ can be written as the union of $H_q (\mathcal{S}_{\boldsymbol{n}} \cap V_q)$ for a countable family of points $q\in Z_{\boldsymbol{n}}$. Hence, it suffices to study the structure of $H_q (\mathcal{S}_{\boldsymbol{n}} \cap V_q)$ for a given $q\in Z_{\boldsymbol{n}}$.\\

If $\sigma=(\sigma_1,\dots,\sigma_k)\in\{+1,-1\}^k$, we define
$$\calS^\sigma_{\boldsymbol{n}}  := \left\{ ((x,\xi_1),\dots,(x,\xi_k))\in\calS_{\boldsymbol{n}} \text{ such that $\forall i,j\in\{1,\dots,k\}$, we have $\sigma_i\xi_i=\sigma_j\xi_j$} \right\}.$$
Then, $\calS^\sigma_{\boldsymbol{n}}$ is a smooth submanifold of $\calS_{\boldsymbol{n}}$, and $\calS_{\boldsymbol{n}}\setminus\left(\cup_{\sigma\in\{+1,-1\}^k}\calS^\sigma_{\boldsymbol{n}}\right)$ is open in $\calS_{\boldsymbol{n}}$. We shall write
$$Z_{\boldsymbol{n}}^\sigma := (\Psi^k)^{-1} \left(\calS^\sigma_{\boldsymbol{n}} \cap (\mathcal{R}(\Sigma))^k\right).$$

The following lemma can be deduced from a simple geometric argument which we will give at the end of the proof.
\begin{tcolorbox}
\begin{lemma}\label{lem:RanProj}
We have
\begin{align*}
&\forall \sigma\in\{+1,-1\}^k,~  \forall q \in Z_{\boldsymbol{n}}^\sigma,~~ &\dim(\textup{Ker}(d_{\Psi^k(q)} H_q)\cap T_{\Psi^k(q)}\calS_{\boldsymbol{n}}^\sigma)=1,\\
&\forall q \in Z_{\boldsymbol{n}}  \setminus\left(\cup_{\sigma\in\{+1,-1\}^k}Z_{\boldsymbol{n}}^\sigma\right),  ~~&\dim(\textup{Ker}(d_{\Psi^k(q)} H_q)\cap T_{\Psi^k(q)}\calS_{\boldsymbol{n}})=0.
\end{align*}
\end{lemma}
\end{tcolorbox}
It follows from Lemma \ref{lem:RanProj} and from the fact that $\calS_{\boldsymbol{n}}$ is a smooth submanifold of $(S^*X)^k$ of codimension $kd-1$ that $Z'_{\boldsymbol{n}}= \pi \left( (\Psi^k)^{-1}(\calS_{\boldsymbol{n}}\cap (\mathcal{R}(\Sigma))^k) \right)$ is a countable union of submanifolds of $(S^*X|_\Sigma)^k$ of codimension at least $kd-1$. As a consequence, $Z''_{\boldsymbol{n}}$ is also a countable union of submanifolds of $(T^*\Sigma)^k$ of codimension at least $kd-1$. Now, by the multijet transversality theorem (see Theorem 4.13, Chap. 2 of \cite{GG}) the set $E_{\boldsymbol{n}}$ of $u\in C^\infty(\Sigma)$ such that the section $(du,\dots,du):(y_1,\dots,y_k)\mapsto (d_{y_1}u,\dots,d_{y_k}u)$ of $(T^*\Sigma)^k\simeq T^*(\Sigma^k)$ is transversal to $Z''_{\boldsymbol{n}}$, a residual in $C^\infty(\Sigma)$. But since $\textup{dim}(\Sigma^k)=k(d-1)<kd-1=\textup{codim}(Z''_{\boldsymbol{n}})$ (as we have assumed that $k\geq 2$), transversality in this case implies that the range of $(du,\dots,du)$ never intersects $Z''_{\boldsymbol{n}}$.
The result follows.
\end{proof}
\begin{proof}[Proof of Lemma \ref{lem:RanProj}]
The map $H_q$ is clearly a submersion, and is invariant by the action of the geodesic flow on each component. If we write $p= ((x_1,\xi_1),\dots,(x_k,\xi_k)) = \Psi^k(q)$, then we see that
$\textup{Ker}(d_p H_q)$ is generated by the $(0,\dots,(\xi_j,0),\dots,0)$ where the factor $(\xi_j,0)$ corresponds to the $j$-th factor of $S^*X$ and $\xi_j$ acts on the horizontal part of the tangent bundle $T(S^*X)$. In particular, $\textup{Ker}(d_p H_q)$ is a subspace of the horizontal subspace of $T_p(S^*X)^k$.  On the other hand, if $p\in\calS_{\boldsymbol{n}}$, the intersection of $T_p\calS_{\boldsymbol{n}}$ with the horizontal subspace of $T_p(S^*X)^k$ is exactly the diagonal $D_p$ of this horizontal subspace (i.e., the set of $((v,0),\dots,(v,0))$ where $v$ ranges over all of $T_{x_1}X$). Thus, the corank of $H_q$ at $p$ is the dimension of the space $\textup{Ker}(d_p H_q)\cap D_p$. This space is trivial except when the $\xi_j$'s are all colinear, in which case it is exactly the line generated by $((\xi_1,0),\dots,(\xi_1,0))$. The statement follows.
\end{proof}

\begin{appendix}
\section{A review of semi-classical analysis}\label{sec:semi-classique}
In this section we review some basic definitions from semiclassical analysis and state Egorov's theorem. Let $Y$ be a smooth $d$-dimensional Riemannian manifold. In all the paper, $Y$ is either the compact manifold $X$ or its universal cover $\widetilde{X}$.\\

\textcolor{black}{We shall use the class $S^{comp}(T^*Y)$ of symbols $a\in C_c^\infty(T^*Y)$,
which may depend on $h$, but whose semi-norms and supports are all bounded
independently of $h$. }

\textcolor{black}{Using coordinate charts, and the standard Weyl quantization on $\mathbb{R}^d$ as in \cite[\S 14.2]{Zworski_2012}, we may associate to each symbol in $a\in S^{comp}(T^*Y)$ an operator $\mathrm{Op}_h(a)$, acting on functions of $Y$. We thus obtain a quantization map}
\begin{equation*}
\text{Op}_h:S^{comp}(T^*Y)\longrightarrow \Psi^{comp}_h(Y).
\end{equation*}
This construction is not intrinsic. However, as explained in \cite[Theorem 14.2]{Zworski_2012}, the principal
symbol map
\begin{equation*}
\sigma_h : \Psi^{comp}_h (Y)\longrightarrow S^{comp}(T^*Y)/
h S^{comp}(T^*Y)
\end{equation*} is intrinsic, and we have
\begin{equation*}\sigma_h(A\circ B) = \sigma_h (A) \sigma_h(B)
\end{equation*}
and
\begin{equation*}
\sigma_h\circ \text{Op}_h: S^{comp}(T^* Y) \longrightarrow S^{comp} (T^*Y) /h
S^{comp}(T^*Y)
\end{equation*}
is the natural projection map. The operators in $\Psi_h^{comp}(Y)$ are always bounded independently of $h$ when acting on $L^2(Y)$, as explained in \cite[Theorem 14.2]{Zworski_2012}.\\

Let $(f_h)$ be a bounded family in $L^2(Y)$, and let $\nu$ be a measure on $T^*Y$. We \textcolor{black}{say} that $(f_h)$ has a semi-classical measure $\nu$ (which is then unique), if, for any $a\in C_c^\infty(T^*Y)$, we have
\[
\left\langle f_h, \text{Op}_h(a) f_h \right\rangle \underset{h\to 0}{\longrightarrow} \int_{T^*Y} a(x,\xi) \mathrm{d}\nu(x,\xi)\, .
\]
Let $f_h:= e^{\frac{i}{h} \phi(x)} a(x)$, with $a\in C_c^\infty(Y)$ and $\phi$ a smooth function defined in a neighbourhood of the support of $a$. As explained in \cite[\S 5.1, Example 2]{Zworski_2012}, $(f_h)$ has a semi-classical measure, which is given by
\begin{equation*}
\nu = |a(x)|^2 \mathrm{d}x \delta_{\xi= \partial \phi(x)}\, .
\end{equation*}
More generally, if  $f_h(x) = \sum_{j=1}^N e^{\frac{i}{h} \phi_j(x)} a_{j}(x)$, with $\partial \phi_j(x) \neq \partial \phi_{j'}(x)$ for all $j\neq j'$ and all $x$ in the support of both $a_j$ and $a_{j'}$,  then a similar proof (using non-stationary phase to show that the non-diagonal terms are negligible) implies that $f_h$ has a semi-classical measure, which is given by
\begin{equation}\label{eq:semi-classical_measure}
\nu = \sum_{j=1}^N |a_j(x)|^2 \mathrm{d}x \delta_{\xi= \partial \phi_j(x)}.
\end{equation}

The following result, known as Egorov's theorem, whose proof can be found in \cite[\S 15]{Zworski_2012}, says that, when considering semi-classical measures, the classical and quantum evolutions commute. \textcolor{black}{Recall that $U_h(t):=e^{ith\frac{\Delta}{2}}:L^2(X)\rightarrow L^2(X)$ is the semiclassical Schrödinger propagator.}

\begin{tcolorbox}[breakable, enhanced]
\begin{theorem}[Egorov's theorem, Theorem 15.2 of \cite{Zworski_2012}]\label{t:Egorov}
Let $(f_h)$ be a bounded family in $L^2(Y)$ having a semi-classical measure $\nu_0$. Then, for any $t\in \R$, the family $\left( U_h(t) f_h \right)$ has a semi-classical measure $\nu_t$, which is given by 
\[
\nu_t = \Phi^t_* \nu_0\, .
\]
\end{theorem}
\end{tcolorbox}

\section{Construction of monochromatic phases}\label{sec:AppMono}

In this section we describe how to construct the phase $\phi_u$ starting from a function $u\in\calC^T(\Sigma)$ as announced in section \ref{sec:IntroMono}. Let $\Sigma\subset X$ be an embedded orientable simply connected hypersurface. Let us denote by $\nu$ a vector field defined on $\Sigma$ such that for each $y\in\Sigma$, $\nu(y)$ has unit norm and is orthogonal to $T_y\Sigma$. Recall the definition of $\mathcal{C}^T(\Sigma)$ from \eqref{eq:calC_trans}. Consider the following first order PDE, where the unkown is a smooth function $\psi:\Omega\rightarrow\R$ defined on an open neighbourhood $\Omega$ of $\Sigma$ in $X$:

\begin{equation}\label{e:eikonal_PDE}
\begin{cases}
|\partial\psi| &=1;\\
\psi|_\Sigma &=u;\\
\partial\psi|_\Sigma &=v_u\, .
\end{cases}
\end{equation}
According to the following lemma Equation \eqref{e:eikonal_PDE} always admits solutions and these satisfy a certain (weak) uniqueness property.
\begin{tcolorbox}[breakable, enhanced]
\begin{lemma}\label{l:existence_uniqueness}
Let $u\in\calC^T(\Sigma)$. Then, there exists a neighbourhood $\Omega$ of $\Sigma$ in $X$ and a function $\psi\in C^\infty(\Sigma)$ which solves \eqref{e:eikonal_PDE}. Moreover, if $(\psi_1,\Omega_1)$ and $(\psi_2,\Omega_2)$ are two such solutions then $\psi_1$ and $\psi_2$ coincide on some open subset $\Omega_3\subset\Omega_1\cap\Omega_2$.
\end{lemma}
\end{tcolorbox}
\begin{proof}
The lemma follows by the method of characteristics. Indeed, for the PDE \eqref{e:eikonal_PDE}, the hypersurface $\Sigma$ is non-characteristic (in the sense of \cite[(36) p.106]{Evans}). Therefore, we can apply the method of characteristics and deduce local existence and uniqueness of the solution near each point of $\Sigma$ (see \cite{Evans}). Then, piecing together local solutions using the uniqueness, we obtain a global solution near $\Sigma$. Moreover, given two solutions $(\psi_1,\Omega_1)$ and $(\psi_2,\Omega_2)$ of \eqref{e:eikonal_PDE}, by local uniqueness, for each $x\in\Sigma$, there exists $U_x\subset\Omega_1\cap\Omega_2$ on which they coincide. In particular, $\Omega_3=\cup_{x\in\Sigma} U_x$ is a neighbourhood of $\Sigma$ on which they coincide.
\end{proof}
To better understand the solutions to Equation \eqref{e:eikonal_PDE}, we check that solutions to this equation satisfy a property that makes them simple do describe in terms of the initial condition.
\begin{tcolorbox}[breakable, enhanced]
\begin{lemma}\label{l:description_eikonal}
Let $U\subset X$ be an open subset and let $\phi\in C^\infty(U)$ be such that $|\partial\phi|=1$. Then, if $x\in U$ and $v=\partial\phi(x)$, we have, for each $t\in \R$ close enough to $0$, $\phi(\Phi^t(x,v))=\phi(x)+t$.
\end{lemma}
\end{tcolorbox}
%{\color{black} Recall that $X$ is equipped with a Riemannian metric $g$, which, in particular, induces a Levi-Civita connection on $X$.}
\begin{proof}
For simplicity, we assume that $\phi(x)=0$. Since $|\partial_x\phi|=1$, the level set $\Sigma=\phi^{-1}(0)$ is a smooth hypersurface orthogonal to the vector field $\partial\phi$. For each $t\in\R$, let $\gamma_t:\Sigma\rightarrow X$ be the geodesic flow starting from $\Sigma$ in the direction $\partial\phi|_\Sigma$. There exists a neighbourhood $\tilde{W}\subset\Sigma\times\R$ of $\Sigma\times\{0\}$ such that the map $\gamma:(y,t)\mapsto\gamma_t(y)$ from $W$ into $X$ is a diffeomorphism onto its image $W$. Let $\psi\in C^\infty(W)$ be defined by $\psi(\gamma_t(y))=t$. Let us show that $\psi$ solves \eqref{e:eikonal_PDE}. To do so, first note that the level sets of $\psi$ are of the form $\gamma_t(U)$ where $U\subset\Sigma$ is an open subset. Next, note that for each $y\in\Sigma$ and each $t\in\R$, $\dot{\gamma}_t(y)\perp T_y(\gamma_t(\Sigma))$. Indeed, this is true for $t=0$ and, if we write $\nabla$ for the Levi-Civita connection on $X$ {\color{black}induced by $g$}, for each vector field $v$ on $\Sigma$,
\begin{align*}
\frac{d}{dt} g(d\gamma_t(v),\dot{\gamma}_t)&=g(\nabla_{\frac{d}{dt}}(d\gamma_t v),\dot{\gamma}_t)+g(d\gamma_t v,\nabla_{\frac d dt}\dot{\gamma}_t)\textup{ by compatibility}\\
&=g(\nabla_{\frac{d}{dt}}(d\gamma_t v),\dot{\gamma}_t)\textup{ since }\nabla_{\frac d dt}\dot{\gamma}_t=0\\
&=g(\nabla_v(\dot{\gamma}_t),\dot{\gamma}_t)\textup{ by symmetry}\\
&=(1/2)d\left[g(\dot{\gamma}_t,\dot{\gamma}_t)\right]v\textup{ by compatibility}\\
&=0\textup{ because }g(\dot{\gamma}_t,\dot{\gamma}_t)=1\, .
\end{align*}
Since we also have $|\dot{\gamma}_t|=1=\frac{d}{dt}\psi(\gamma_t)$, we deduce that for each $(y,t)\in\tilde{W}$, $\partial\psi(\gamma_t(y))=\dot{\gamma}_t(y)$. In particular, $\psi$ solves \eqref{e:eikonal_PDE} as announced. Moreover, $\psi|_\Sigma=\phi|_\Sigma$ so, by Lemma \ref{l:existence_uniqueness}, they coincide near $\Sigma$. In particular, $\phi(\Phi^t(x,\partial_x\phi))=\psi(\Phi^t(x,\partial_x\psi))=\psi(\gamma_t(x))=t$ for $|t|$ small enough, and the proof is over.
\end{proof}
\begin{remark}
Though we never use this property in the article, note that Lemma \ref{l:description_eikonal} allows us to describe a maximal choice of $\Omega_3$ from Lemma \ref{l:existence_uniqueness}. More precisely, if $(\psi_1,\Omega_1)$ and $(\psi_2,\Omega_2)$ are two such solutions then $\psi_1$ and $\psi_2$ coincide on 
\[
\Omega_1\cap \Omega_2\cap\{\Phi^t(x,\partial\psi_1(x))\, :\, x\in\Sigma,\, t\in\R\}\, .
\]
Indeed, notice that the characteristic curves of \eqref{e:eikonal_PDE} are gradient lines of $\psi_1$ and $\psi_2$. But by Lemma \ref{l:description_eikonal}, these are geodesics started at points of the form $(x,\partial\psi_j(x))$ for $j=1,2$ and $x\in\Omega_j$. In particular, taking $x\in\Sigma$, these gradient lines coincide for both $\psi_1$ and $\psi_2$ and the two functions must coincide along them as long as they are well defined.
\end{remark}
\end{appendix}

%\bibliography{biblicomplete}
%\bibliographystyle{amsalpha}
%\bibliographystyle{abbrv}
\end{document}